\documentclass[conference]{IEEEtran}
\usepackage{cite}

\ifCLASSINFOpdf
\else
\fi
\usepackage{amsmath}
\usepackage{hyperref}
\usepackage{csquotes}
\usepackage{amssymb}
\usepackage{cleveref}
\usepackage{nicefrac}
\usepackage{tikz}
\usetikzlibrary{automata}
\usetikzlibrary{positioning}
\usetikzlibrary{arrows}
\usepackage{bbm}
\usepackage{booktabs}
\usepackage{stmaryrd}
\usepackage{thm-restate}
\usepackage{soul}
\usepackage{mdframed}
\usepackage{footnote}
\newenvironment{mdframedwithfoot}
{   
    \savenotes
    \begin{mdframed}
        \stepcounter{footnote}
        
    }
    {
    \end{mdframed}
    \spewnotes
}
\usepackage{makecell}
\usepackage{adjustbox}

\usepackage{etoolbox}
\newtoggle{arxiv}
\toggletrue{arxiv}

\begin{document}
    
\IEEEoverridecommandlockouts 
\IEEEpubid{\makebox[\columnwidth]{979-8-3503-3587-3/23/\$31.00~ \copyright2023 IEEE \hfill} \hspace{\columnsep}\makebox[\columnwidth]{ }}
    
%
\title{On Certificates, Expected Runtimes, and Termination in Probabilistic Pushdown Automata
\thanks{This work has been supported by the DFG RTG 2236 (UnRAVeL), the EU’s Horizon 2020 research and innovation programme under the Marie Skłodowska-Curie grant agreement No 101008233 (MISSION), and the ERC Advanced Grant No 787914 (FRAPPANT).}
}

\author{\IEEEauthorblockN{Tobias Winkler and Joost-Pieter Katoen}
\IEEEauthorblockA{Software Modeling and Verification Group\\
RWTH Aachen University, Aachen, Germany\\
Email: \{tobias.winkler,katoen\}@cs.rwth-aachen.de}
}


%


\maketitle


\begin{abstract}
Probabilistic pushdown automata (pPDA) are a natural operational model for a variety of recursive discrete stochastic processes.
In this paper, we study \emph{certificates} -- succinct and easily verifiable proofs -- for upper and lower bounds on various quantitative properties of a given pPDA.
We reveal an intimate, yet surprisingly simple connection between the existence of such certificates and the expected time to termination of the pPDA at hand.
This is established by showing that certain intrinsic properties, like the spectral radius of the Jacobian of the pPDA’s underlying polynomial equation system, are directly related to expected runtimes.
As a consequence, we obtain that there always exist easy-to-check proofs for \emph{positive almost-sure termination}: does a pPDA terminate in finite expected time?
\end{abstract}


%

\newtheorem{example}{Example}
\newtheorem{theorem}{Theorem}
\newtheorem{lemma}{Lemma}
\newtheorem{corollary}{Corollary}
\newtheorem{proposition}{Proposition}
\newtheorem{definition}{Definition}
\newtheorem{remark}{Remark}
\newtheorem{contribution}{Contribution}
\newenvironment{proof}{\begin{IEEEproof}}{\end{IEEEproof}}

\renewcommand{\vec}{\boldsymbol}

\newcommand{\eps}{\varepsilon}

\newcommand{\allsmaller}{\prec}
\newcommand{\allgreater}{\succ}

\newcommand{\nats}{\mathbb{N}}
\newcommand{\reals}{\mathbb{R}}
\newcommand{\nonnegreals}{\reals_{\geq 0}}
\newcommand{\exnonnegreals}{\overline{\reals}_{\geq 0}}
\newcommand{\rats}{\mathbb{Q}}
\newcommand{\nonnegrats}{\rats_{\geq 0}}

\renewcommand{\P}{\mathit{PTIME}}
\newcommand{\NP}{NP}
\newcommand{\PSPACE}{\mathit{PSPACE}}

\newcommand{\norm}[1]{\|{#1}\|}
\newcommand{\maxnorm}[1]{\norm{#1}_{\infty}}

\newcommand{\idmat}{I} 
\newcommand{\genmat}{A} 
\newcommand{\specrad}[1]{\rho({#1})} 

\newcommand{\polyring}[2]{#1[#2]}

\newcommand{\gfvar}{\mathsf{z}}
\newcommand{\gfring}[2]{#1[[#2]]}
\newcommand{\gf}{g}
\newcommand{\gfdiff}[1]{\partial{#1}}
\newcommand{\pdiff}[2]{\partial_{#1}{#2}} 
\newcommand{\gfsubs}[2]{{#1}\big\rvert_{\gfvar=#2}}
\newcommand{\gfmc}[2]{g^{#1}_{#2}}

\newcommand{\sring}{\mathcal{S}}
\newcommand{\scarrier}{A}
\newcommand{\selem}{a}
\newcommand{\selemm}{b}
\newcommand{\selemmm}{c}
\newcommand{\splus}{\oplus}
\newcommand{\ssum}{\bigoplus}
\newcommand{\stimes}{\otimes}
\newcommand{\szero}{0}
\newcommand{\sone}{1}
\newcommand{\sringinit}{(\scarrier, \splus, \stimes, \szero, \sone)}
\newcommand{\natord}{\leq}

\newcommand{\sringexnonnegreals}{\mathcal{R}_{\geq 0}^{\infty}}
\newcommand{\sringgf}[2]{#1[[#2]]}

\newcommand{\autom}{\mathfrak{A}}
\newcommand{\amat}{M} 
\newcommand{\astates}{S}
\newcommand{\afinal}{f}
\newcommand{\behav}[2]{||#1;#2||}

\newcommand{\mc}{\mathfrak{M}}
\newcommand{\mcstates}{\astates}
\newcommand{\mcmat}{\amat}
\newcommand{\termmc}[2]{#1_{#2}}

\newcommand{\tprobmc}[2]{\mathbb{P}^{#1}(\lozenge #2)}
\newcommand{\rvrt}[1]{\mathcal{R}_{#1}} 
\newcommand{\rtmc}[2]{\mathbb{E}^{#1}(\rvrt{#2})} 
\newcommand{\rtmcuncond}[2]{\mathbb{E}^{#1}(\rvrt{#2} \land \lozenge #2)}
\newcommand{\rtmccond}[2]{\mathbb{E}^{#1}(\rvrt{#2} \mid \lozenge #2)}

\newcommand{\pda}{\Delta}
\newcommand{\pdastates}{Q}
\newcommand{\abstack}{\Gamma}
\newcommand{\pdatrans}{\delta}
\newcommand{\pdainit}{(\pdastates,\abstack,\pdatrans)}
\newcommand{\pdafinal}{\afinal}
\newcommand{\trans}[3]{#1 \overset{#2}{\to} #3}
\newcommand{\transin}[4]{#1 \overset{#2}{\to_{#4}} #3}
\newcommand{\semautom}[1]{\autom_{#1}}
\newcommand{\semmat}[1]{T_{#1}}

\newcommand{\triple}[3]{[#1#2{\mid}#3]}
\newcommand{\vtriple}[3]{\langle #1#2{\mid}#3 \rangle}

\newcommand{\etriple}[3]{\mathsf{E}\triple{#1}{#2}{#3}}
\newcommand{\vetriple}[3]{\mathsf{E}\vtriple{#1}{#2}{#3}}

\newcommand{\ert}[2]{\mathsf{ert}[#1#2]}
\renewcommand{\vert}[2]{\mathsf{ert}\langle#1#2\rangle}
\newcommand{\ertmat}[1]{M_{#1}}

\newcommand{\sys}[1]{\vec{#1}} 
\newsavebox{\foobox}
\newcommand{\slantbox}[2][.0]
{%
    \mbox
    {%
        \sbox{\foobox}{#2}%
        \hskip\wd\foobox
        \pdfsave
        \pdfsetmatrix{1 0 #1 1}%
        \llap{\usebox{\foobox}}%
        \pdfrestore
    }%
}
\newcommand{\syscl}[1]{\slantbox{$\mathbbm{f}$}} 
\newcommand{\lfp}[1]{\mu #1} 
\newcommand{\gfp}[1]{\nu #1} 
\newcommand{\pdasys}[1]{\sys{f}_{#1}}
\newcommand{\pdasyscl}[1]{\mathbbm{f}_{#1}}
\newcommand{\jac}[1]{#1'} 
\newcommand{\jacat}[2]{\jac{#1}(#2)} 

\newcommand{\todo}[1]{\textcolor{red}{#1}}

\section{Introduction}
\label{sec:intro}

Probabilistic pushdown automata (pPDA)~\cite{DBLP:conf/lics/EsparzaKM04} are known as a universal model for various kinds of recursive discrete stochastic processes.
This includes stochastic context-free grammars, probabilistic programs with procedures and recursion~\cite{DBLP:conf/lics/OlmedoKKM16}, as well as many random walks on $\mathbb{N}$.
pPDA are equivalent to \emph{recursive Markov chains}~\cite{DBLP:conf/stacs/EtessamiY05,DBLP:journals/jacm/EtessamiY09}.

The pPDA considered in this paper are similar to standard pushdown automata.
The difference is that a pPDA does not read symbols from an input word, but rather determines its next configuration stochastically based on fixed distributions over the states' outgoing transitions.
A pPDA can thus be seen as a finite-state discrete-time Markov chain with an \emph{unbounded} stack.
See \Cref{fig:simple} for an example.

A plethora of quantitative verification problems for pPDA has been studied, including termination probabilities, expectations and tail bounds of key random variables such as running time~\cite{DBLP:conf/icalp/BrazdilKKV11,DBLP:journals/jcss/BrazdilKKV15}, rewards/costs~\cite{DBLP:conf/lics/EsparzaKM05}, memory consumption~\cite{DBLP:conf/fsttcs/BrazdilEK09},
satisfaction probabilities of linear-time~\cite{DBLP:conf/lics/EsparzaKM04,DBLP:conf/stacs/BrazdilKS05,DBLP:conf/qest/YannakakisE05,DBLP:journals/lmcs/KuceraEM06,DBLP:conf/fossacs/WinklerGK22} and branching-time~\cite{DBLP:conf/stacs/BrazdilKS05,DBLP:journals/lmcs/KuceraEM06,DBLP:journals/jcss/BrazdilBFK14} logics,
as well as probabilistic simulation~\cite{DBLP:journals/iandc/Huang0K19} and bisimulation~\cite{DBLP:journals/acta/BrazdilKS08}
(see~\cite{DBLP:journals/fmsd/BrazdilEKK13} for a detailed overview).

A central issue in most of the above problems is to compute a certain set of probabilities, denoted $\triple p Z q$ for states $p,q$ and stack symbol $Z$.
$\triple p Z q$ stands for the probability that the pPDA eventually reaches $q$ with empty stack assuming it is initiated in $p$ with stack content $Z$~\cite{DBLP:conf/lics/EsparzaKM04}.
It is well-known that the numbers $\triple p Z q$ are characterized as a \emph{least fixed point} solution $\lfp\sys f$ of a \emph{positive polynomial system} (PPS) $\sys f$~\cite{DBLP:reference/hfl/Kuich97,DBLP:conf/lics/EsparzaKM04,DBLP:conf/stacs/EtessamiY05}; and it is thus not surprising that in general, $\triple p Z q$ may be irrational and cannot easily be represented in a closed form using radical expressions~\cite{DBLP:journals/jacm/EtessamiY09}.
On the other hand, many problems mentioned in the previous paragraph remain \emph{decidable} by encoding $\triple p Z q$ in the existential first-order theory of real arithmetic, which is in $\PSPACE$.

To cope with pPDA in practice, and to obtain better bounds than $\PSPACE$, lots of research has focused on computing the probabilities $\triple p Z q$ \emph{approximately}.
The resulting algorithms employ variants of \emph{Newton iteration} for the systems generated by pPDA and subclasses thereof~\cite{DBLP:conf/stacs/EtessamiY05,DBLP:journals/jacm/EtessamiY09,DBLP:journals/jacm/StewartEY15},
as well as for arbitrary polynomial systems with non-negative coefficients~\cite{DBLP:conf/stoc/KieferLE07,DBLP:conf/stacs/EsparzaKL08,DBLP:journals/siamcomp/EsparzaKL10,tacas}.

\begin{figure}[t]
    \centering
    \begin{tikzpicture}[thick,initial text=$Z$,node distance=20mm]
    \node[state,initial] (p) {$p$};
    \node[state,right=of p] (q) {$q$};
    \draw[->] (p) -- node[auto] {$(\nicefrac 1 4, Z, \eps)$} (q);
    \draw[->] (p) edge[loop above] node[left] {$(\nicefrac 1 4, Z, ZZ)$} (p);
    \draw[->] (p) edge[loop below] node[left] {$(\nicefrac 1 2, Z, \eps)$} (p);
    \draw[->] (q) edge[loop right] node[auto] {$(1, Z, \eps)$} (q);
    \end{tikzpicture}
    \caption{An example probablistic pushdown automaton (pPDA)~\cite{tacas} with states $\pdastates = \{p,q\}$ and stack alphabet $\abstack = \{Z\}$. A transition $(a,Z,\alpha)$ means that with probability $a$, $Z$ is popped from the stack and $\alpha$ is pushed.}
    \label{fig:simple}
\end{figure}

\subsubsection*{Our Contributions}

In this paper, we are concerned with the following fundamental question:
Does there exist a \emph{simple proof} -- also called \emph{certificate} -- that a claimed bound on a given quantitative pPDA property indeed holds?
By \emph{simple} we mean that the proof should be machine-checkable by a verifier with limited resources; in particular, it should be checkable in $\P$.
Moreover, the check should be conceptually simple such that the verifier is straightforward to implement, and ideally amenable to formal verification.
We primarily restrict attention to bounds on the probabilities $\triple p Z q$ as those are at the heart of solving many other problems, see above.

Our motivation for studying certificates is twofold.
First, certificates can help increasing the trust one can put in the outputs of a (possibly very complex) analysis engine or probabilistic model checker.
Indeed, a model checker that provides certificates for its results along with a formally verified certificate checker is as trustworthy as a fully formally verified model checker -- the latter is usually not in reach for tools relevant in practice.
This is exactly the idea of \emph{certifying algorithms}~\cite{DBLP:journals/csr/McConnellMNS11}.
In fact, the authors of \cite{DBLP:journals/csr/McConnellMNS11} go one step further and argue that
\begin{quote}
    \enquote{
        [...] \emph{certifying algorithms are much superior to non-certifying algorithms and} [...] \emph{for complex algorithmic tasks only certifying algorithms are satisfactory.}
    }
\end{quote}

Our second, more theoretical motivation is that the study of certificates may of course also spawn new algorithms for pPDA, because an algorithm for finding a certificate (or proving absence thereof) for a given hypothesis is in particular an algorithm deciding that hypothesis.

To the best of our knowledge, certificates for pPDA have only been considered explicitly in the very recent paper~\cite{tacas}.
There, the basic idea is to exploit the fairly general rule
\begin{align}
    \label{eq:inductionintro}
    \sys f(\vec u) \leq \vec u
    ~\implies~
    \lfp\sys f \leq \vec u
\end{align}
which is true for all \emph{monotonic} functions $\sys f \colon L \to L$ with least fixed point $\lfp\sys f$, and $(L, \leq)$ a complete lattice.
\eqref{eq:inductionintro} follows from the Knaster-Tarski fixed point theorem and is in particular applicable to the PPS $\sys f$ arising from pPDA where the condition $\sys f(\vec u) \leq \vec u$ is reasonably simply to check.
Applying \eqref{eq:inductionintro} in practice, however, is not trivial because it is not clear if suitable, effectively representable (i.e., rational-valued) $\vec u$ even exist.
The key contributions of this paper are to extend the ideas of~\cite{tacas} in the following ways (also see \Cref{fig:overview}):

\begin{mdframedwithfoot}
\begin{contribution}
    \label{contrib:1}
    We characterize the existence of \emph{rational}-valued certificates for arbitrarily tight \emph{upper} bounds on $\triple p Z q$ in terms of the \emph{expected runtime}\footnote{Also called running time or time to termination. \enquote{Expected runtime} is a term inspired by~\cite{DBLP:conf/lics/OlmedoKKM16}.} of the given pPDA $\pda$.
    In particular, we show that such certificates exist if the expected runtime \emph{conditioned on termination} is finite, a property called \emph{cPAST} in this paper.
    We also obtain similar though slightly more complex results for \emph{lower} bounds.
    Moreover, we show that the bit-complexity of the rational numbers occurring in the certificates can be bounded in terms of the expected runtime and the minimum non-zero probability $\triple p Z q$.
\end{contribution}
\end{mdframedwithfoot}

Our second contribution concerns the termination problem of pPDA: Does a given pPDA eventually reach the empty stack with probability one?
This is crucial for numerous applications, e.g. model checking temporal logics~\cite{DBLP:conf/fossacs/WinklerGK22}, or ensuring that a probabilistic program has a well-defined posterior distribution~\cite[Sec.~5.3]{DBLP:journals/corr/abs-1809-10756}.

\begin{mdframed}
\begin{contribution}
    We prove that \emph{positive almost-sure termination} (PAST; termination with probability 1 in finite expected time) can \emph{always} be certified (if it holds).
    This relies on Contribution~\ref{contrib:1} and a novel characterization of the expected runtimes in general pPDA.
    The main insight, which is of independent interest beyond the topic of certificates, is that \emph{PAST can be proved using only sufficiently tight \emph{upper} bounds} on the probabilities $\triple p Z q$.
    This allows deciding termination of a large class of pPDA (all but those that are AST and not PAST) by practically feasible, iterative approximation algorithms.
\end{contribution}
\end{mdframed}

\begin{figure*}[t]
    \centering
    \begin{adjustbox}{max width=\textwidth}
        {\renewcommand{\arraystretch}{0.1}
        \begin{tabular}{l l l l p{4cm}}
            \toprule
            \emph{What}? & \emph{Certificate data} & \emph{Verification condition(s)} & \makecell[l]{\emph{Certified property}} & \emph{Remarks} \\
            \midrule
            \makecell[l]{Upper bounds \\ \small \Cref{thm:certificates}} & $\vec u \in \nonnegrats^{\pdastates \times \abstack \times \pdastates}$ & $\sys f(\vec u) \leq \vec u$ & \makecell[l]{$\forall pZq \in \pdastates{\times}\abstack{\times}\pdastates \colon$ \\ $\quad \triple{p}{Z}{q} \leq \vec u_{pZq}$}  & Certificate exists if $\pda$ is cPAST\\
            \midrule
            \makecell[l]{Lower bounds \\ \small \Cref{thm:lower}} & $\vec l, \vec u \in \nonnegrats^{\pdastates \times \abstack \times \pdastates}$ &  \makecell[l]{
                $\vec l \leq \sys f(\vec l)$, $\sys f(\vec u) \allsmaller \vec u$ \\ $\vec l \leq \vec u$}
            & \makecell[l]{$\forall pZq \in \pdastates{\times}\abstack{\times}\pdastates \colon$ \\ $\quad \vec l_{pZq} \leq \triple{p}{Z}{q} < \vec u_{pZq}$} & Certificate exists if $\pda$ is cPAST\\
            \midrule
            \makecell[l]{PAST \\ \small \Cref{thm:certpast}} & \makecell[l]{$\vec u \in \nonnegrats^{\pdastates \times \abstack \times \pdastates}$ \\ $\vec r \in \rats_{\geq 1}^{\pdastates\times\abstack}$} & \makecell[l]{$\sys f(\vec u) \leq \vec u$ \\ $\ertmat{\pda}(\vec u)\vec r + \vec 1 \leq \vec r$} & \makecell[l]{$\forall pZ \in \pdastates{\times}\abstack \colon$ \\ $\quad \ert{p}{Z} \leq \vec r_{pZ}$} & Certificate exists if and only if $\pda$ is PAST\\
            \bottomrule
        \end{tabular}}
    \end{adjustbox}
    \caption{$\pda = \pdainit$ is a pPDA with fundamental PPS $\sys f$. $\ertmat{\pda}(\vec u)$ is a matrix whose coefficients can be read off $\pda$ and $\vec u$. The first two rows apply more generally to least fixed point of \emph{arbitrary} PPS (existence is then guaranteed if the system $\sys f$ is non-singular).}
    \label{fig:overview}
\end{figure*}

\subsubsection*{Proof techniques}
As a primary technical vehicle, we prove that the spectral radius $\rho$ of the underlying system's Jacobi matrix $\sys f'(\lfp\sys f)$ is closely related to the (conditional) expected runtimes (assuming $\sys f$ has been previously cleaned-up by removing \enquote{unnecessary} variables).
This relation is expressed via the simple (tight) inequality $\rho \leq 1 - C^{-1}$, where $C$ is the maximal conditional expected runtime ocurring in the pPDA at hand ($C = \infty$ is possible).
We prove this using results about general PPS from~\cite{DBLP:journals/siamcomp/EsparzaKL10} as well as standard facts about non-negative matrices.
Another important ingredient are  \emph{weighted PDA} over semirings~\cite{DBLP:reference/hfl/Kuich97}.

\subsubsection*{Related work}
Expected runtimes and, more generally, expected rewards in pPDA are studied in~\cite{DBLP:conf/lics/EsparzaKM05,DBLP:conf/icalp/BrazdilKKV11,DBLP:journals/fmsd/BrazdilEKK13,DBLP:journals/jcss/BrazdilKKV15}.
However, these works focus mostly on expectations \emph{conditioned} on reaching the empty stack while being in a certain state.
In our paper we need such conditional expectations as well; but additionally, we also consider the \emph{un}conditional expected runtime of pPDA, which is crucial for proving PAST.
More specifically, our \Cref{thm:ertsys} has -- to the best of our knowledge -- not appeared in previous works.
A pPDA is PAST iff the equation system from \Cref{thm:ertsys} has a non-negative solution.
This is in contrast to e.g. the equation system from \cite[Ex.~3.3]{DBLP:conf/lics/EsparzaKM04} which has a solution iff the expected runtime \emph{conditioned on termination} is finite.
However, this can be the case even if the termination probability is $< 1$.
\Cref{thm:ertsys} is actually more in the spirit of \cite{DBLP:conf/icalp/EtessamiWY08}, \cite[Ch.~4]{DBLP:phd/ethos/Wojtczak09}, where \emph{1-exit recursive simple stochastic games} with positive rewards are considered, a model that corresponds to a \emph{stateless} variant of PDA controlled by two players and a probabilistic environment.
As we consider general pPDA with states, our result does not easily follow from theirs.

The expected \emph{memory consumption} of a pPDA~\cite{DBLP:conf/fsttcs/BrazdilEK09} is also related to its runtime as finite runtime implies finite memory usage.
Apart from this, there seems to be no other known connection between the two.

It is shown in~\cite{DBLP:journals/jacm/EtessamiY09} that the spectral radius of the Jacobi matrix $\sys f'(\vec 1)$ is related to the termination probability in the case of 1-exit RMC, which are equivalent to stateless pPDA.
However, expected runtimes are not considered in~\cite{DBLP:journals/jacm/EtessamiY09}.
In fact, the connection between expected runtimes and spectral properties of the Jacobi matrix seems to have been ignored entirely in previous works.

Apart from~\cite{tacas}, certificates for various other probabilistic and quantitative systems are studied, including Markov decision processes~\cite{ovi,DBLP:conf/tacas/FunkeJB20}, and timed automata~\cite{DBLP:conf/tacas/WimmerM20}.
A major technical difference to these works is that they do not need to consider non-linear arithmetic.

\cite{DBLP:conf/lics/OlmedoKKM16} introduces a deductive approach to prove bounds on, amongst other properties, expected runtimes of (possibly infinite-state) pPDA represented as probabilistic recursive programs.
Our certificates for PAST are loosely inspired by theirs, and we use some of their terminology.

There also exist generalizations of pPDA for incorporating non-determinism~\cite{DBLP:conf/icalp/EtessamiY05} or higher-order recursion~\cite{DBLP:journals/lmcs/KobayashiLG20}.
We leave certification of these models for future work.

The recent paper~\cite{DBLP:journals/pacmpl/ChistikovMS22} studies certificates for (non)-reachability in non-probabilistic PDA and related problems.

\subsubsection*{Paper outline}
The upcoming \Cref{sec:example} presents a small example in detail.
We then summarize necessary preliminary results and definitions in \Cref{sec:prelims}.
\Cref{sec:upper} and \Cref{sec:lower} present our results for certificates for upper and lower bounds, respectively.
\Cref{sec:cpast} establishes the fundamental connection between the Jacobi matrix and the expected runtimes.
Certificates for PAST are introduced in \Cref{sec:past}.
Bounds on the size of certificates are given in \Cref{sec:complexity}.
We conclude in \Cref{sec:conclusion}.

\section{Illustrative Example}
\label{sec:example}

Consider the pDPA $\pda$ in \Cref{fig:simple} on \cpageref{fig:simple}.
To simplify the presentation, we write $x_p = \triple p Z p$ and $x_q = \triple p Z q$, i.e., $x_p$ and $x_q$ are the probabilities that $\pda$ terminates in state $p$ and $q$, respectively, assuming $\pda$ is initiated in $p$ with stack content $Z$ (such \emph{configurations} are denoted $pZ$ in the following).
As mentioned earlier, $x_p$, $x_q$ are the component-wise least solution of a polynomial equation system, called \emph{fundamental system} of $\pda$ in this paper (see \Cref{thm:pdafundamentalsys} on page~\pageref{thm:pdafundamentalsys} for the general equations).
In our example:
\begin{align}
    x_p &~=~ \frac 1 4 x_p^2 + \frac 1 2 \label{eq:simple1} \\
    x_q &~=~ \frac 1 4 x_p x_q + \frac 1 4 x_q + \frac 1 4 \label{eq:simple2}
\end{align}
The exact $x_p$ and $x_q$ are the \emph{irrational} values $x_p = 2 - \sqrt 2 \approx 0.59$, $x_q = \sqrt 2 - 1 \approx 0.41$.
We stress that in general, it is impossible to express these numbers with radicals (square roots, cubic roots, etc)~\cite{DBLP:journals/jacm/EtessamiY09} -- in the example this only works because the polynomials have small degree.
Therefore one is usually interested in proving (rational) bounds on these probabilities.

\subsubsection*{Certificates for upper bounds (\Cref{sec:upper})}
To prove that e.g. $u_p = \frac 3 5, u_q = \frac 1 2$ are valid upper bounds for $x_p$ and $x_q$, it is sufficient (but not necessary) that \eqref{eq:simple1} and \eqref{eq:simple2} hold with \emph{inequality} \enquote{$\geq$} when replacing $x_p$ by $u_p$ and $x_q$ by $u_q$.
This is indeed true for our example values, e.g. $\frac 3 5 \geq \frac 1 4 \cdot (\frac 3 5)^2 + \frac 1 2$.
Note that these checks can be done in polynomial time (in the size of $\pda$ and the bit-complexity of the rational numbers $u_p, u_q$) and are easy to implement -- thinking of them as a certificate is thus justified.
However, it is not clear if \emph{rational-valued} bounds satisfying the constraints can always be found.
In our paper we characterize a large class of pPDA where rational certificates are guaranteed to exist.

\subsubsection*{Certificates for lower bounds (\Cref{sec:lower})}
The above idea can be extended to lower bounds as well, but this is slightly more involved.
We will prove that $l_p = \frac 4 7$ and $l_q = \frac 2 5$ are valid lower bounds.
Our theory only allows proving lower bounds in conjunction with upper bounds, so we will just reuse the upper bound certificate $u_p = \frac 3 5, u_q = \frac 1 2$.
\begin{itemize}
    \item First, we check that \eqref{eq:simple1} and \eqref{eq:simple2} are satisfied with inequality \enquote{$\leq$} (reverse direction as for upper bounds) when replacing $x_p$ by $l_p$ and $x_q$ by $l_q$.
    This holds.
    \item Next we check that $l_p \leq u_p$ and $l_q \leq u_q$. This also holds.
    \item Finally, we verify that \eqref{eq:simple1} and \eqref{eq:simple2} hold with \emph{strict} inequality \enquote{$>$} when replacing $x_p$ by $u_p$ and $x_q$ by $u_q$. 
\end{itemize}
We conclude (correctly) that $\frac 4 7 \leq 2 - \sqrt 2$ and $\frac 2 5 \leq \sqrt 2 - 1$.

\subsubsection*{Positive almost-sure termination (\Cref{sec:past})}
It is intuitively clear that our example pPDA terminates (reaches the empty stack) with probability 1, but this property is difficult to prove in general.
With the lower bound certificates seen above we can only conclude that the termination probability is at least $\frac 4 7 + \frac 2 5 \approx 0.97$.
However, our theory also admits certificates for PAST, which implies termination with probability \emph{exactly} $1$.
Similar to the lower bound certificates, a certificate for PAST consists of an upper bound certificate $u_p, u_q$ as well as numbers $r_p$, $r_q$ which are upper bounds on the \emph{expected runtime} starting in $pZ$ and $qZ$, respectively.
In our example, let $u_p = \frac 3 5, u_q = \frac 1 2$ as before, and let $r_p = \frac{45}{14}$, and $r_q = 1$.
To \emph{verify} the certificate we just have to check the inequalities
\begin{align}
    r_p &~\geq~ 1 + \frac 1 4 (r_p + u_p r_p) + \frac 1 4 (r_p + u_q r_q) \\
    r_q &~\geq~ 1
\end{align}
These inequalities can be easily read off from the transitions of $\pda$, just like the fundamental system \eqref{eq:simple1} and \eqref{eq:simple2}.
We find that the inequalities are satisfied and conclude that $\pda$ is PAST.
In fact, the expected time to termination starting in $pZ$ is at most $r_p = \frac{45}{14} \approx 3.21$.

\emph{This is all that's necessary to prove positive almost-sure termination} -- and thus in particular \emph{termination with probability 1} -- and we show in this paper that this method is \emph{complete}, i.e., it works for arbitrary pPDA that are PAST, even if only rational-valued bounds are allowed.


\section{Preliminaries}
\label{sec:prelims}

\subsection{General Notation and Markov Chains}
$\exnonnegreals = \nonnegreals \cup \{\infty\}$ denotes the set of \emph{extended} non-negative reals.
Arithmetic and comparison with the $\infty$-symbol is done as usual: $\forall a \in \exnonnegreals \colon a + \infty = \infty + a = \infty$, $0 \cdot \infty = \infty \cdot 0 = 0$, $\forall a \in \exnonnegreals \setminus \{0\} \colon a \cdot \infty = \infty \cdot a = \infty$, and $\forall a \in \nonnegreals \colon a < \infty$.
With these conventions, $(\exnonnegreals, +, \cdot, 0, 1)$ is a commutative semiring.
Throughout the paper, we often encounter diverging monotonic sequences $a_0 \leq a_1 \leq \ldots$ in $\nonnegreals$; in such cases, we define $\lim_{i \to \infty} a_i = \infty$ (and the resulting $\infty$ symbol may well be subject to subsequent arithmetic operations obeying the above conventions).

We denote vectors by boldface letters such as $\vec v$.
All vectors in this paper can be thought of as column vectors.
A vector with all entries equal to some $a \in \exnonnegreals$ is written $\vec a$.
For instance, $\vec 0$ and $\vec 1$ denote the all-zeros and all-ones vectors, respectively.
We always index a vector's components by the elements of a countable, but otherwise arbitrary index set $J$, and we write e.g. $\vec v \in \nonnegreals^J$ to indicate that $\vec v$ is a non-negative real vector with $|J|$ components.
The individual components of a vector $\vec v$ are referred to as $\vec v_i$, $i \in J$.
Given vectors $\vec u, \vec v \in \exnonnegreals^J$, we define the following \emph{comparison operators}:
$\vec u \leq \vec v$ if $\forall i\in J \colon \vec u_i \leq \vec v_i$,
$\vec u < \vec v$ if $\vec u \leq \vec v$ and $\vec u \neq \vec v$, and
$\vec u \allsmaller \vec v$ if $\forall i\in J \colon \vec u_i < \vec v_i$.
Note that \enquote{$\leq$} is a partial order on $\exnonnegreals^J$.

For matrices we adopt the same index set convention as for vectors:
A (square) matrix $A\in \reals^{J \times J}$ is indexed by pairs of elements from $J$.
For $(i,j) \in J\times J$, $A_{i,j}$ denotes the entry in the row and column indexed by $i$ and $j$, resp.

In this paper, a \emph{discrete-time Markov chain} (DTMC) $\mc = (\mcstates, \mcmat)$ consists of a  finite or countably infinite set $\mcstates$ of states and a row-stochastic transition probability matrix $\mcmat \in [0,1]^{\mcstates\times\mcstates}$.
Intuitively, for $s,t \in \mcstates$, $M_{s,t}$ is the probability that $\mc$ moves in one step from state $s$ to state $t$.
To make $\mc$ a proper stochastic process, it is necessary to fix an initial state $s_{init} \in \mcstates$; however, for the sake of flexibility, we deliberately do not make $s_{init}$ an explicit part of $\mc$.
We often consider DTMC with a distinguished target (or final) set $F \subseteq \mcstates$.
We then write $\mcmat_F$ for the matrix resulting from $\mcmat$ by setting all rows with indices in $F$ to zero.

\subsection{Runtime and Termination: Different Flavours}

\subsubsection{Termination Moments}
Let $\mc = (\mcstates, \mcmat)$ be a DTMC with $F \subseteq \mcstates$ a distinguished set of target states.
Assign probability $p(\pi)$ to each non-empty path $\pi = s_0\ldots s_l \in \mcstates^+ = \mcstates^* \setminus \{\eps\}$ in the natural way: $p(s_0\ldots s_l) = \prod_{i=0}^{l-1} \mcmat_{s_i, s_{i+1}}$.
The \emph{length} of path $\pi = s_0\ldots s_l$ is defined as $|\pi| = l$.
By convention, a path of length zero is assigned probability 1.
For a given initial state $s \in \mcstates \setminus F$, let $Paths_{s \to F}^\mc = \{s\}.(S \setminus F)^*.F$ be the set of finite paths from $s$ to $F$ that visit $F$ exactly once (namely at the end; for $s\in F$ we define $Paths_{s \to F}^\mc = \{s\}$).
Such paths are called \emph{proper paths between $s$ and $F$} in the following.
It is easy to see that
$
    p \colon Paths_{s \to F}^\mc \to [0,1]
$
is a \emph{sub-probability} distribution, i.e., $\sum_{\pi \in Paths_{s \to F}^\mc} p(\pi) \leq 1$ (strict inequality is possible because there might be a positive probability to never reach $F$ from $s$).

\begin{definition}[Termination Moments]
    \label{def:terminationmoments}
    Let $\mc = (\mcstates, \mcmat)$ be a DTMC.
    For initial state $s \in \mcstates$, target set $F \subseteq \mcstates$ and integer $k \geq 0$, we define the \emph{$k$th termination moment} $\mathsf{E}^k[s{\mid}F]_{\mc}$ as the expected value of $|\pi|^k$ measured wrt.\ the \emph{sub}-distribution $p \colon Paths_{s \to F}^\mc \to [0,1]$.
    Formally:
    \begin{align}
        \mathsf{E}^k[s{\mid}F]_{\mc}
        ~:=~
        \sum_{\pi\in Paths_{s \to F}^\mc} p(\pi) \cdot |\pi|^k
        ~\in~
        \exnonnegreals
    \end{align}
    with the convention that $0^0 = 1$.
    We write just $\mathsf{E}^k[s{\mid}F]$ if $\mc$ is clear from the context.
    For $f \in \mcstates$, we also write $\mathsf{E}^k[s{\mid}f]$ instead of $\mathsf{E}^k[s{\mid}\{f\}]$.
\end{definition}

For readers familiar with the \emph{weakest pre-expectation calculus}~\cite{DBLP:series/mcs/McIverM05} we mention that $\mathsf{E}^k[s{\mid}F]_{\mc}$ can be expressed as $\mathsf{wp}\llbracket {P_{\mc}} \rrbracket(\mathtt{c}^k)(s)$ where $P_{\mc}$ is a suitable probabilistic program encoding the DTMC $\mc$, and $\mathtt{c}$ is a special \enquote{counter variable} tracking the number of computation steps in $P_{\mc}$.

\Cref{def:terminationmoments} involves summing over all (countably many) proper paths between a given initial state $s$ and target set $F$.
A standard observation in the theory of DTMC (e.g.~\cite{kulkarni1995modeling}) is that the sum of the probabilities of all proper paths \emph{of length $k \geq 0$} between $s$ and $\{f\}$ is given by $(\mcmat_f^{k})_{s,f}$,
and thus the sum of the probabilities of all proper paths of \emph{any} finite length between $s$ and $\{f\}$ is equal to $\sum_{k=0}^{\infty}(\mcmat_f^k)_{s,f}$.
The following lemma allows to express such infinite sums as the solution of a linear equation system; it is an instance of Kleene iteration.

\begin{lemma}[{Kleene; e.g.~\cite{DBLP:reference/hfl/Kuich97}}]
    \label{thm:explicitsollinsys}
    Let $J$ be a countable index set, and $\genmat \in \nonnegreals^{J \times J}$, $\vec b \in \nonnegreals^J$.
    Then the vector
    \begin{align}
        \sum_{k = 0}^{\infty} \genmat^k \vec b
        ~=~
        (\sum_{k = 0}^{\infty} \genmat^k) \vec b
        ~\in~
        \exnonnegreals^J
    \end{align}
    is the least $\exnonnegreals$-valued solution of the linear equation system $\vec x = \genmat \vec x + \vec b$.
\end{lemma}

Let $\mc = (\mcstates,\mcmat)$ be a DTMC with target $F \subseteq \mcstates$.
In this paper, we only need the first two termination moments:
\begin{itemize}
    \item The $0$th termination moment $\mathsf{E}^0[s{\mid}F] = [s{\mid}F]$ is just the probability to reach $F$ from a given initial state $s \in \mcstates$.
    By \Cref{thm:explicitsollinsys}, the numbers $\{[s{\mid}F] \mid s \in \mcstates\}$ constitute the least solution $\vec \mu^0 \in \nonnegreals^\mcstates$ of the linear equation system
    \begin{align}
        \vec \mu^0
        ~=~
        M_F \vec \mu^0 + \vec \chi (F)
        ~,
    \end{align}
    where $\vec \chi (F)$ is the characteristic $\{0,1\}$-vector of $F$, i.e., $\vec \chi (F)_s = 1$ iff $s \in F$.
    \item Similarly, it can be shown that the $1$st termination moments $\mathsf{E}^1[s{\mid}F] = \mathsf{E}[s{\mid}F]$ are the least solution $\vec \mu^1 \in \exnonnegreals^S$ of the linear equation system
    \begin{align}
        \label{eq:1stmoments}
        \vec \mu^1
        ~=~
        M_F \vec \mu^1 + (\vec \mu^0 - \vec\chi(F))
        ~.
    \end{align}
\end{itemize}

\subsubsection{Conditional Expected Runtimes}
Let $\mc = (\mcstates, \mcmat)$ be a DTMC with target $F \subseteq \mcstates$.
If for a given state $s \in \mcstates$ we have that $[s{\mid}F] > 0$, then we call the quantity $\frac{\mathsf{E}[s{\mid}F]}{[s{\mid}F]}$ the \emph{conditional expected runtime} of $\mc$ wrt.\ to target $F$ and initial state $s$.
Intuitively, $\frac{\mathsf{E}[s{\mid}F]}{[s{\mid}F]}$ is the number of states visited until reaching $F$ from $s$, \emph{conditioned} on the fact that $F$ is eventually reached.


\subsubsection{(Unconditional) Expected Runtimes}

Intuitively, the expected runtime of a DTMC wrt.\ to initial state $s$ and target set $F$ is the total expected number of states visited until reaching $F$ from $s$.
In particular, if $F$ is not reached with probability $1$, then the expected runtime is $\infty$.
This is the key difference to the 1st termination moment, i.e., $\mathsf{E}[s{\mid}F]$ may be finite even if $[s{\mid}F] < 1$.

It is possible to define the expected runtimes formally as the expectation of a random variable on the probability space of \emph{infinite} paths induced by $\mc$, but for our purposes, the following simpler characterization suffices:

\begin{definition}[Expected runtimes]
    \label{def:ert}
    Let $\mc = (\mcstates, \mcmat)$ be a DTMC with a target set $F \subseteq \mcstates$.
    The \emph{expected runtime} of $\mc$ wrt.\ target $F$ and initial state $s \in \mcstates$ is defined as
    \begin{align}
        \label{eq:ertsum}
        \ert{s{\mid}F}{}_\mc
        ~=~
        \Big( \sum_{k=0}^{\infty} \mcmat_{F}^k \vec ( \vec 1 - \vec\chi(F)) \Big)_{s}
        ~\in~
        \exnonnegreals
        ~.
    \end{align}
    We write just $\ert{s{\mid}F}{}$ if $\mc$ is clear from the context.
\end{definition}
Note that by \Cref{thm:explicitsollinsys}, the values $\ert{s{\mid}F}{}$ constitute the least solution $\vec r \in \exnonnegreals^S$ of the linear equation system
\begin{align}
    \label{eq:lesert}
    \vec r
    ~=~
    \mcmat_{F} \vec r + ( \vec 1 - \vec\chi(F))
    ~.
\end{align}

It follows from \eqref{eq:1stmoments} that if $\vec \mu^0 = \vec 1$ (i.e., $\mc$ is \emph{almost-surely terminating}, see below), then $\ert{s{\mid}F}{} = \mathsf{E}[s{\mid}F]$.
In general, however, we may have $\ert{s{\mid}F}{} > \mathsf{E}[s{\mid}F]$.

Another observation that will be crucial in the proof of \Cref{thm:ertsys} is that $\ert{s{\mid}F}{}$ is the sum of the probabilities of \emph{all} finite paths starting in $s$ that never visit $F$, i.e., paths of the form $\{s\}.(\mcstates \setminus F)^*$.
This follows from \eqref{eq:ertsum} and the discussion preceding \Cref{thm:explicitsollinsys}.

\subsubsection{Probabilistic Termination}

\begin{definition}[AST, PAST, and cPAST]
    \label{def:terminationNotions}
    Given a DTMC $\mc = (\mcstates,\mcmat)$ with target $F \subseteq \mcstates$, we define:
    \begin{itemize}
        \item $\mc$ is \emph{almost-surely terminating} (AST) wrt.\ $F$ if $\forall s \in \mcstates$, we have $[s{\mid}F] = 1$, i.e., $F$ is reached from every state $s$ with probability $1$.
        \item $\mc$ is \emph{positively AST} (PAST) wrt.\ $F$ if $\forall s \in \mcstates$, $\ert{s{\mid}F}{} < \infty$.
        \item $\mc$ is conditionally PAST (cPAST) wrt.\ F if $\forall s \in \mcstates$ such that $[s{\mid}F] > 0$, it holds that $\frac{\mathsf{E}[s{\mid}F]}{[s{\mid}F]} < \infty$.
    \end{itemize}
\end{definition}
Note that, since $[s{\mid}F] = 0$ implies $\mathsf{E}[s{\mid}F] = 0$, a DTMC is cPAST iff the 1st termination moments are all finite.

The notion of cPAST, although less well known than AST and PAST, has been considered in~\cite{DBLP:conf/lics/EsparzaKM05,DBLP:journals/jcss/BrazdilKKV15} (the term \enquote{cPAST} has not been used before though).
Intuitively, cPAST means that termination either occurs \enquote{fast} (in finite expected time), or not at all.

Note that we consider \emph{universal} (\enquote{for all initial states}) definitions of AST, PAST, and cPAST.
This is mostly to simplify the presentation, as otherwise we would have to treat uninteresting special cases such as unconnected DTMC seperately.
It is possible to extend our results to the setting of a distinguished initial state of interest by conducting an appropriate reachability analysis in advance.

\begin{restatable}{lemma}{terminationRelations}
    \label{thm:terminationRelations}
    The relationship between AST, PAST, and cPAST in a DTMC with a distinguished target set is as follows:
    
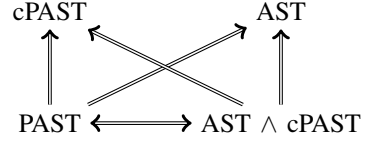
\begin{center}
        \begin{tikzpicture}[node distance = 10mm and 20mm]
        \node[] (cPAST) {cPAST};
        \node[below=of cPAST] (PAST) {PAST};
        \node[right=of cPAST] (AST) {AST};
        \node[below=of AST] (ASTcPAST) {AST $\land$ cPAST};
        \draw[double,->] (PAST) -- (AST);
        \draw[double,->] (PAST) -- (cPAST);
        \draw[double,->] (ASTcPAST) -- (AST);
        \draw[double,->] (ASTcPAST) -- (cPAST);
        \draw[double,<->] (ASTcPAST) -- (PAST);
        \end{tikzpicture}
    \end{center}
    The omitted implications do not hold in general.
\end{restatable}
\begin{proof}
    We give a self-contained proof for the implication PAST $\implies$ AST as this is central to our work.
    See~\cite{DBLP:conf/lics/OlmedoKKM16,dal2020probabilistic} for alternative proofs and \iftoggle{arxiv}{\Cref{app:terminationRelations}}{\cite{arxiv}} for the other (straightforward) implications.
    
    To see that PAST implies AST, first note that for $n \geq 0$,
    \begin{align}
        \label{eq:forshowingastpast}
        M_F^n(\vec 1 - \vec\chi (F)) ~+~ \sum_{k=0}^n \mcmat_F^k \vec \chi(F)
        ~=~
        \vec 1
    \end{align}
    which follows by induction using that $\mcmat_F\vec 1 + \vec\chi(F) = \vec 1$.
    Now suppose that $\mc$ is not AST, i.e., $\vec\mu^0 < \vec 1$.
    By \Cref{thm:explicitsollinsys}, $\vec\mu^0$ can be written explicitly as $\sum_{k = 0}^\infty \mcmat_F^k \vec\chi(F)$.
    From \eqref{eq:forshowingastpast} and $\vec\mu^0 < \vec 1$, it follows that $\lim_{n \to \infty} M_F^n(\vec 1 - \vec\chi (F)) > \vec 0$.
    As the expected runtimes are defined as $\vec r = \sum_{k = 0}^\infty \mcmat_F^k ( \vec 1 - \vec\chi(F))$ (see~\ref{eq:ertsum}), it follows that at least one component of $\vec r$ is $\infty$.
    Thus, $\mc$ is not PAST.
\end{proof}

\subsection{Weighted and Probabilistic PDA}

We now define pushdown automata (PDA) with non-negative \emph{rational} transition weights, an instance of PDA with semiring-valued weights as described e.g. in~\cite{DBLP:reference/hfl/Kuich97}.

A weighted pushdown automaton (weighted PDA) is a triple $\pda = (\pdastates, \abstack, \pdatrans)$,
where $\pdastates$ is a \emph{finite} set of \emph{states},
$\abstack$ is a \emph{finite stack alphabet},
and $\pdatrans \subseteq \pdastates \times \abstack \times \nonnegrats \times \pdastates \times \abstack^{\leq 2}$ is a \emph{weighted transition relation} (here, $\abstack^{\leq 2} = \{\eps\} \cup \abstack \cup \abstack^2$).
Intuitively, a transition $(p,Z,\selem,q,\alpha) \in \delta$ means that if $\pda$ is in state $p$ and the stack content is $Z\beta$ for some $\beta \in \abstack^*$ (i.e., $Z$ is on top of the stack), then there exists a transition with weight $\selem$ to the configuration where $\pda$ is in state $q$ and the stack content is $\alpha\beta$ (i.e., $Z$ is popped and $\alpha$ is pushed on the stack).
Instead of $(p,Z,a,q,\alpha) \in \delta$ we also write $\trans{pZ}{a}{q\alpha}$ if $\pdatrans$ is clear from the context.
In the following, we always write tuples from $\pdastates\times\abstack^*$ without parentheses and separating commata, e.g., we write $pZ$ instead of $(p,Z)$.

A weighted PDA $\pda$ is a \emph{probabilistic} PDA (pPDA) if $\sum_{\trans{pZ}{a}{q\alpha}} a = 1$ holds for all $pZ \in \pdastates{\times}\abstack$~\cite{DBLP:conf/lics/EsparzaKM04}.
We emphasize that these notions of weighted (probabilistic) PDA do \emph{not} recognize weighted (probabilistic) languages; they are just \emph{generative} models assigning certain weights (probabilities) to finite computation traces.

More formally, we define the \emph{semantics} of a weighted PDA $\pda$ in terms of a countably infinite weighted transition system $\mathfrak{T}_{\pda} = (\pdastates\times\abstack^*, \semmat{\pda})$, i.e., $\semmat{\pda}$ is a $\nonnegrats$-valued square martix indexed by the \emph{configurations} $\pdastates{\times}\abstack^*$ of $\pda$ (which are the states of $\mathfrak{T}_{\pda}$).
$\semmat{\pda}$ is defined as follows:
For all $\trans{pZ}{\selem}{q\alpha}$ and $\beta \in \abstack^*$, $\semmat{\pda}(pZ\beta, q\alpha\beta) = a$, and all other entries of $\semmat{\pda}$ are zero.
Note that if $\pda$ is a pPDA, then $\mathfrak{T}_{\pda}$ can be seen as a DTMC with target set $F = \{p\eps \mid p \in\pdastates\}$ because the states of the form $p\eps$ have no outgoing transitions in $T_{\pda}$.

We extend the notion of termination moments (Def.~\ref{def:terminationmoments}) to the weighted transition systems $\mathfrak{T}_{\pda}$ induced by a weighted PDA $\pda$ in the natural way.
For reasons explained below, we mostly consider $0$th moments of the form $\triple{p}{Z}{q\eps} _{\pda}$, i.e., the sum of the weights of all finite paths from $pZ$ to $q\eps$.
In the following, we omit $\eps$ and $\pda$ from the notation whenever this causes no confusion and write just $\triple p Z q$.
In a pPDA, $\triple p Z q$ is the probability to reach $q\eps$ from initial configuration $pZ$.
In the context of pPDA, the probabilities $\triple p Z q$ are also called \emph{return probabilities}~\cite{DBLP:conf/fossacs/WinklerGK22}.

We will also consider 1st termination moments of the form $\etriple{p}{Z}{q\eps}_{\pda}$ in pPDA (but never in general weighted PDA).
For the sake of brevity, those are denoted as $\etriple p Z q$.


\begin{theorem}[Fundamental system of PDA{\cite{DBLP:reference/hfl/Kuich97},\cite{DBLP:conf/lics/EsparzaKM04}}]
    \label{thm:pdafundamentalsys}
    Let $\pda = \pdainit$ be a weighted PDA.
    Consider $\exnonnegreals$-valued variable symbols of the form $\vtriple{p}{Z}{q}$, $pZq \in \pdastates{\times}\abstack{\times}\pdastates$.
    Then the quadratic system of equations
    \begin{multline}
        \label{eq:pdasys}
        \vtriple{p}{Z}{q}
        ~=~
        \sum_{\trans{pZ}{a}{q\eps}} a
        ~+~
        \sum_{\trans{pZ}{a}{rY}} a \cdot \vtriple{r}{Y}{q} \\
        ~+~
        \sum_{\trans{pZ}{a}{rYX}} a \cdot  \sum_{t \in \pdastates} \vtriple{r}{Y}{t} \cdot \vtriple{t}{X}{q}
    \end{multline}
    has the (point-wise) least solution $\vtriple{p}{Z}{q} = \triple{p}{Z}{q}_{\pda}$.
\end{theorem}
We refer to~\cite[Sec.~3]{DBLP:journals/lmcs/KuceraEM06} for an intuitive explanation of \eqref{eq:pdasys}.

Finally, we extend the various notions of termination from Def.~\ref{def:terminationNotions} to pPDA $\pda = \pdainit$ and $F = \pdastates \times \{\eps\}$:
\begin{itemize}
    \item $\pda$ is \emph{AST} if the Markov chain $\mathfrak{T}_{\pda}$ is AST wrt.\ $F$, which is the case iff for all $pZ \in \pdastates \times\abstack$, $\triple p Z F _\pda = \sum_{t \in \pdastates} \triple p Z t _\pda = 1$.
    \item $\pda$ is \emph{PAST} if $\mathfrak T _\pda$ is PAST wrt.\ to $F$ which is the case iff for all $pZ \in \pdastates \times\abstack$, $\ert{p}{Z{\mid}F} < \infty$.
    In the context of pPDA, we simply write $\ert{p}{Z}$ instead of $\ert{p}{Z{\mid}F}$ because we will only consider expected runtimes wrt.\ to $F$ as defined above.
    \item $\pda$ is \emph{cPAST} if $\mathfrak T _\pda$ is cPAST wrt.\ to \emph{all} $t\eps \in F$, i.e., for all $pZt \in \pdastates\times\abstack\times\pdastates$, $\frac{\etriple p Z t}{\triple p Z t}$ is either undefined or finite.
\end{itemize}

\subsection{Positive Polynomial Systems (PPS)}

Let $J$ be a finite index set and consider a vector $\vec x = (\vec x_j)_{j \in J}$ of variable symbols.
A \emph{positive polynomial system} (PPS) over the variables $\vec x$ is a vector $\sys f = (\sys f_j)_{j \in J}$ of multivariate polynomials $\sys f_j$ with variables in $\vec x$ and \emph{non-negative coefficients} in $\nonnegreals$.

A PPS $\sys f$ induces an endofunction $\sys f \colon \exnonnegreals^J \to \exnonnegreals^J, \vec u \mapsto \sys f(\vec u) = (\sys f_j(\vec u))_{j \in J}$.
As the coefficients of the polynomials in $\sys f$ are non-negative, the function $\sys f$ is \emph{monotone}, i.e., $\vec u \leq \vec v \implies \sys f(\vec u) \leq \sys f(\vec v)$, where $\leq$ is the element-wise order on $\exnonnegreals^J$.
Since the domain $(\exnonnegreals^J, \leq)$ of $\sys f$ is a \emph{complete lattice}, the Knaster-Tarski fixed point theorem implies that $\sys f$ has a \emph{least fixed point} (lfp) $\lfp\sys f \in \exnonnegreals^J$.
In other words, the polynomial equation system $\vec x = \sys f(\vec x)$ is guaranteed to have a \emph{least} solution $\lfp\sys f \in \exnonnegreals^J$, wrt.\ to $\leq$.
A PPS $\sys f$ is called \emph{feasible} if $\lfp\sys f \allsmaller \vec\infty$, i.e., $\lfp\sys f \in \nonnegreals^J$.

\begin{lemma}[Inductive upper bounds]
    \label{thm:inductiveUpperBounds}
    Let $\sys f$ be a PPS and let $\vec u \geq \vec 0$.
    The following statements hold:
    \begin{itemize}
        \item $\sys f(\vec u) \leq \vec u \implies \lfp\sys f \leq \vec u$
        \item $\sys f(\vec u) \allsmaller \vec u \implies \lfp\sys f \allsmaller \vec u$
    \end{itemize}
\end{lemma}
\begin{proof}
    The first item is an immediate consequence of the Knaster-Tarski fixed point theorem.
    For the second item, note that $\vec u \allgreater \sys f(\vec u) \geq \sys f^1(\vec u) \geq \sys f^2(\vec u) \ldots$ is a monotonic and bounded sequence, hence $\lim_{i \to \infty} \sys f^i(\vec u)$ exists.
    Since $\sys f$ is a continuous function, $\sys f(\lim_{i \to \infty} \sys f^i(\vec u)) = \lim_{i \to \infty} \sys f^{i+1}(\vec u) = \lim_{i \to \infty} \sys f^{i}(\vec u)$, so the latter term is a fixed point of $\sys f$.
    Thus $\lfp\sys f \leq \lim_{i \to \infty} \sys f^{i}(\vec u) \allsmaller \vec u$.
\end{proof}

We often denote PPS in the form of an equation system $\vec x = \sys f(\vec x)$.
In this paper we are mostly concerned with PPS of type \eqref{eq:pdasys}, as well as with \emph{linear} PPS that can be written in the form $\vec x = \genmat \vec x + \vec b$, for some matrix $\genmat \in \nonnegreals^{J \times J}$ and vector $\vec b \in \nonnegreals^J$ (cf.~\Cref{thm:explicitsollinsys}).

A PPS $\sys f$ is called \emph{clean} if $\lfp\sys f \allgreater \vec 0$.
General PPS may be not clean, e.g., $x = 2x$ is not clean.
However, it is possible to compute, in linear time~\cite{DBLP:journals/jacm/EtessamiY09,DBLP:journals/siamcomp/EsparzaKL10}, a \emph{cleaned-up} variant $\syscl f$ of $\sys f$ by identifying and removing variables that are assigned zero in the lfp.
Formally, the cleaned-up system $\syscl f$ satisfies the property that
(i) $\syscl f$ is clean,
(ii) $\syscl f$ is indexed by a subset $J' \subseteq J$,
(iii) for all $j \in J'$, we have $(\lfp{\syscl f})_j = (\lfp\sys f)_j$, and
(iv) for all $j \in J \setminus J'$, $(\lfp{\sys f})_j = 0$.
For example, the cleaned-up version of $(x,y) = (2x, 1 + x + y)$ is just $y = 1 + y$ (whose lfp is $\infty$, i.e., it is not feasible).
\emph{Throughout the paper we consistently use the symbol $\sys f$ for general, not necessarily clean PPS, and we reserve the symbol $\syscl f$ for clean PPS.}

Let $\sys{f}$ be a PPS indexed by $J$.
The \emph{Jacobi matrix} $\jac{\sys{f}}$ of $\sys f$ is a $J {\times} J$-matrix whose entries are the partial derivatives $(\jac{\sys{f}})_{i,j} = \frac{\partial}{\partial \vec x_j} \sys{f}_i$ for $i,j \in J$.
The Jacobi matrix with all entries \emph{evaluated} at point $\vec u \in \nonnegreals^J$ is denoted $\jacat{\sys{f}}{\vec u}$.

For a finite matrix $\genmat \in \reals^{J \times J}$, we recall that $\lambda \in \mathbb{C}$ is an \emph{eigenvalue} if there exists $\vec v \neq \vec 0$ such that $\genmat \vec v = \lambda \vec v$.
The \emph{spectral radius} $\specrad{\genmat}$ is the largest absolute value of the at most $|J|$  eigenvalues of $\genmat$.


The next theorem plays a central role in the analysis of Newton's method for PPS (see also \cite[Lem.~6.5]{DBLP:journals/jacm/EtessamiY09}), and turns out very useful for our purposes as well.

\begin{theorem}[See~{\cite[Theorem 4.1 - 2]{DBLP:journals/siamcomp/EsparzaKL10}}]
    \label{thm:specradleq1}
    For any clean and feasible PPS $\syscl f$ it holds that
    $\forall \vec x \allsmaller \lfp\syscl f \colon \specrad{\jacat{\syscl f}{\vec x}} < 1$, and
    $\specrad{\jacat{\syscl f}{\lfp\syscl f}} \leq 1$.
\end{theorem}
The second statement follows from the first as the eigenvalues of any matrix depend continuously on its entries.

The following variant of Taylor's theorem will be crucial in several proofs throughout the paper:

\begin{lemma}[{Taylor's theorem; see~\cite[Lemma 2.3]{DBLP:journals/siamcomp/EsparzaKL10}}]
    \label{thm:taylor}
    Let $\sys f$ be a PPS.
    For all $\vec u, \vec v \geq \vec 0$,
    \begin{align*}
        \sys f(\vec u) + \jacat{\sys f}{\vec u} \vec v
        ~\leq~
        \sys f(\vec u + \vec v)
        ~\leq~
        \sys f(\vec u) + \jacat{\sys f}{\vec u + \vec v} \vec v
        ~.
    \end{align*}
\end{lemma}
Another variant of \Cref{thm:taylor} that we will use is this:
For $\vec u, \vec v \geq \vec 0$ such that $\vec u - \vec v \geq \vec 0$,
\begin{align*}
    \sys f(\vec u) - \jacat{\sys f}{\vec u} \vec v
    ~\leq~
    \sys f(\vec u - \vec v)
    ~\leq~
    \sys f(\vec u) - \jacat{\sys f}{\vec u - \vec v} \vec v
    ~.
\end{align*}
See \iftoggle{arxiv}{\Cref{app:taylor}}{\cite{arxiv}} for a proof.

\subsubsection*{Non-negative matrices}
We also need some results from the theory of matrices with non-negative entries.
We will mostly apply these results to the Jacobi matrix $\sys f'(\vec u)$ of a PPS $\sys f$, which indeed has non-negative entries for $\vec u \geq \vec 0$.

\begin{theorem}[See~{\cite[Section 10.3, Fact 12e]{nonnegmats}}]
    \label{thm:nonnegmats}
    Let $\genmat \geq 0$ be a finite-dimensional non-negative square matrix.
    The following are equivalent:
    \begin{enumerate}
        \item $\lim_{k \to \infty} \genmat ^k = 0$
        \item $\specrad{\genmat } < 1$
        \item $(\idmat - \genmat)^{-1}$ exists and is non-negative
        \item $\exists \vec v \allgreater \vec 0 \colon \genmat \vec v \allsmaller \vec v$
    \end{enumerate}
    Moreover, under either one of the above conditions, $\sum_{k = 0}^{\infty} \genmat^k = (\idmat - \genmat)^{-1}$.
\end{theorem}

We will moreover use the following results concerning the spectral radius of a non-negative matrix:

\begin{lemma}
    \label{thm:nonnegmatsmiscfacts}
    Let $\genmat \geq 0$ be a non-negative finite square matrix.
    \begin{itemize}
        \item The spectral radius $\specrad{\genmat}$ is an eigenvalue of $\genmat$~\cite[Section 10.3, Fact 1a]{nonnegmats}.
        \item If $0 \leq B \leq \genmat$, then $\specrad{B} \leq \specrad{\genmat}$~\cite[Section 10.3, Fact 7]{nonnegmats}.
        \item Let $c \in \nonnegreals$.
        Then $(\exists \vec u \allgreater \vec 0 \colon \genmat \vec u \leq c \vec u) \implies \specrad{A} \leq c$~\cite[Section 10.3 Fact 6a]{nonnegmats}.
    \end{itemize}
\end{lemma}
\section{Certifying Upper Bounds}
\label{sec:upper}

In this section, we study certificates for upper bounds on the lfp of a general PPS.
We do not (yet) consider pPDA.

\begin{definition}[Non-singular PPS]
    A feasible PPS $\sys f$ is called \emph{non-singular} if the matrix $\idmat - \jacat{\sys{f}}{\lfp{\sys{f}}}$ has an inverse.
\end{definition}

\begin{lemma}
    \label{thm:propsOfJacobi}
    Let $\syscl f$ be feasible, non-singular, and clean.
    Then the conditions from \Cref{thm:nonnegmats} apply to $\jacat{\syscl f}{\lfp\syscl f} \geq 0$.
\end{lemma}
\begin{proof}
    Since $I - \syscl f'(\lfp\syscl f)$ is non-singular, $1$ is \emph{not} an eigenvalue of $\syscl f'(\lfp\syscl f)$.
    By \Cref{thm:specradleq1}, $\specrad{\syscl f'(\lfp\syscl f)} \leq 1$.
    Since the spectral radius of a non-negative matrix is an eigenvalue (\Cref{thm:nonnegmatsmiscfacts}), it follows that $\specrad{\syscl f'(\lfp\syscl f)} < 1$, which is the second condition of \Cref{thm:nonnegmats}.
\end{proof}

In \Cref{thm:propsOfJacobi} it is essential that $\syscl f$ is clean:
Consider e.g. $f \colon x = 2x$.
Then $1 - \jacat{f}{\lfp f} = 1 - 2 = -1$ which has inverse $-1 < 0$.
This is the reason why we often consider clean PPS throughout the paper.

The following generalizes a part of \cite[Lem.~2]{tacas}.

\begin{mdframed}
\begin{theorem}[Certificates for upper bounds on the lfp]
    \label{thm:certificates}
    Let $\syscl f$ be a feasible and clean PPS indexed by $J$.
    Then
    $\syscl f$ is non-singular
    if and only if
    $\exists \vec u \in \reals_{> 0}^J \colon \syscl f(\vec u) \allsmaller \vec u$ (and thus $\lfp\syscl f \allsmaller \vec u$ by \Cref{thm:inductiveUpperBounds}.)
    
    In this case, there exists a \emph{rational} $\vec u$ arbitrarily close to the lfp $\lfp\syscl f$, i.e., $\forall \eps > 0 \colon \exists \vec u \in \rats_{> 0}^J \colon \syscl f(\vec u) \allsmaller \vec u \land \maxnorm{\lfp\syscl f - \vec u} < \eps$.
    %
\end{theorem}
\end{mdframed}

\begin{proof}[Proof of \Cref{thm:certificates}]
    Assume $\syscl f$ is non-singular.
    By the right inequality of \Cref{thm:taylor}, $\syscl f(\lfp\syscl f + \vec v) \leq \lfp\syscl f + \jacat{\syscl f}{\lfp\syscl f + \vec v}\vec v$ for all $\vec v \geq \vec 0$.
    By \Cref{thm:propsOfJacobi} and \Cref{thm:nonnegmats}, $\jacat{\syscl f}{\lfp\syscl f}\vec v \allsmaller \vec v$ for some $\vec v \allgreater \vec 0$.
    Since the function $\vec x \mapsto \jacat{\syscl f}{\vec x}$ is continuous, we can choose $\delta > 0$ s.t.\ $\jacat{\syscl f}{\lfp\syscl f + \delta \vec v}\vec v \allsmaller \vec v$ still holds.
    Now
    \begin{align}
        \syscl f(\lfp\syscl f + \delta \vec v)
        ~\leq~
        \lfp\syscl f + \delta \jacat{\syscl f}{\lfp\syscl f + \delta \vec v} \vec v
        ~\allsmaller~
        \lfp\syscl f + \delta \vec v
        ~,
    \end{align}
    so the claim holds with $\vec u := \lfp\syscl f + \delta \vec v \allgreater \vec 0$.
    Furthermore, note that we can choose $\delta > 0$ sufficiently small s.t.\ $\maxnorm{\vec u - \lfp\syscl f} < \eps$ for any given $\eps > 0$.
    Since $\syscl f$ is continuous, $\syscl f(\vec u') \allsmaller \vec u'$ holds for all $\vec u'$ in an open ball around $\vec u$; such a ball contains rational vectors.
    
    For the other direction assume that there exists $\vec u \allgreater \vec 0$ with $\syscl f(\vec u) \allsmaller \vec u$.
    By \Cref{thm:inductiveUpperBounds}, $\lfp \syscl f \allsmaller \vec u$, and so $\vec u = \lfp\syscl f + \vec v$ for some $\vec v \allgreater \vec 0$.
    The left inequality of \Cref{thm:taylor} yields:
    \begin{align}
        \label{eq:proofNonsingIffCert}
        \lfp\syscl f + \jacat{\syscl f}{\lfp\syscl f} \vec v
        ~\leq~
        \syscl f(\lfp\syscl f + \vec v)
        ~\allsmaller~
        \lfp\syscl f + \vec v
        ~.
    \end{align}
    As $\syscl f$ is feasible, $\lfp\syscl f \allsmaller \vec\infty$, and thus $\jacat{\syscl f}{\lfp\syscl f} \vec v \allsmaller \vec v$ by \eqref{eq:proofNonsingIffCert}.
    The latter statement is the last condition from \Cref{thm:nonnegmats}, so $\idmat - \jacat{\syscl f}{\lfp\syscl f}$ has an inverse (which is also non-negative).
\end{proof}
We show in \Cref{sec:cpast} that in the case of PPS associated with pPDA, non-singularity is equivalent to cPAST.

\section{Certifying Lower Bounds}
\label{sec:lower}

In this section we show that it is also possible to certify lower bounds in general PPS.
The contents of this section are not crucial for understanding the subsequent sections.

Recall that the simple induction principle from \Cref{thm:inductiveUpperBounds} does \emph{not} work for lower bounds on the lfp.
Indeed, if $f \colon L \to L$ is a monotonic function and $L$ a complete lattice, then we only have the \enquote{co-induction} rule
\[
    \forall l \in L \colon\quad  l \leq f(l) ~\implies~ l \leq \gfp f ~,
\]
where $\gfp f$ is the \emph{greatest} fixed point of $f$.
As a consequence, the implication $l \leq f(l) \implies l \leq \lfp f$ holds for all $l \in L$ if and only if $f$ has a \emph{unique} fixed point.
However, a PPS $\sys f$ may have multiple fixed points in general.
For example, the one-variable PPS
$x = f(x) = x^2 + \frac{1}{5}$
satisfies $1 \leq f(1) = \frac 6 5$, but $\lfp f = \frac 1 2 - \frac{1}{2 \sqrt 5} < 1$.

We circumvent this problem by providing a co-inductive point $\vec l$ along with a (strict) inductive upper bound $\vec u$:

\begin{mdframed}
\begin{theorem}[Certificates for lower bounds on the lfp]
    \label{thm:lower}
    Let $\sys f$ be a PPS indexed by $J$.
    Then for all $\vec l, \vec u \in \nonnegreals^J$,
    \[
        \vec l \leq \sys{f}(\vec l) \land \vec l \leq \vec u \land \sys{f}(\vec u) \allsmaller \vec u
        ~\implies~
        \vec l \leq \lfp{\sys f} \allsmaller \vec u
        ~.
    \]
    Moreover, if $\sys f$ is feasible, non-singular and clean, then for all $\eps > 0$ there exist $\vec l, \vec u \in \nonnegrats^J$ satisfying the above conditions such that $\maxnorm{\vec l - \vec u} \leq \eps$.
\end{theorem}
\end{mdframed}
\begin{proof}
    Let $\vec l, \vec u$ be such that $\vec l \leq \sys{f}(\vec l), \vec l \leq \vec u$, and $\sys{f}(\vec u) \leq \vec u$.
    The first step of the proof is to construct a fixed point $\tilde{\vec u} \allsmaller \vec u$.
    Let $\tilde{\vec u} = \lim_{n \to \infty} \sys f^n(\vec u) \allsmaller \vec u$.
    By monotonicity, we have $\vec l \leq \sys f^n(\vec l) \leq \sys f^n(\vec u)$ for all $n \geq 0$, so $\vec l \leq \tilde{\vec u}$.
    Moreover, $\tilde{\vec u}$ is a fixed point of $\sys f$ since the latter is continuous, so by the variant of Taylor's theorem (below \Cref{thm:taylor}),
    \begin{align}
        &\tilde{\vec u} = \sys f(\tilde{\vec u}) = \sys f(\vec u - (\vec u - \tilde{\vec u})) \leq \sys f(\vec u) - \sys f'(\tilde{\vec u})(\vec u - \tilde{\vec u}) \\
        \implies &\sys f'(\tilde{\vec u})(\vec u - \tilde{\vec u}) \leq \sys f(\vec u) - \tilde{\vec u} \allsmaller \vec u - \tilde{\vec u} \label{eq:reviewer1}
        ~.
    \end{align}
    The last inequality implies that $\specrad{\sys f'(\tilde{\vec u})} < 1$ by \Cref{thm:nonnegmats}.
    We will need this further below.

    We now show that $\lfp\sys f = \tilde{\vec u}$.
    Towards contradiction assume that $\tilde{\vec u} > \lfp\sys f$, i.e., $\tilde{\vec u} - \lfp\sys f > \vec 0$.
    By the right inequality of Taylor's theorem,
    \begin{align}
        \label{eq:lbprooftaylor}
        \sys{f}(\tilde{\vec u})
        =
        \sys{f}(\lfp\sys f + (\tilde{\vec u}-\lfp\sys f))
        \leq
        \lfp\sys f + \jacat{\sys f}{\tilde{\vec u}}(\tilde{\vec u}-\lfp\sys f)
        .
    \end{align}
    From \eqref{eq:lbprooftaylor} and $\tilde{\vec u} = \sys f(\tilde{\vec u})$ it follows that
    \begin{align}
        \label{eq:lowerboundseq2}
        \tilde{\vec u} - \lfp\sys f
        ~=~
        \sys{f}(\tilde{\vec u}) - \lfp\sys f
        ~\leq~
        \jacat{\sys f}{\tilde{\vec u}}(\tilde{\vec u}-\lfp\sys f)
        ~.
    \end{align}
    Using \eqref{eq:lowerboundseq2} and the fact that $\sys f'(\tilde{\vec u})$ is non-negative, it follows that $\tilde{\vec u}- \lfp\sys f \leq \sys f'(\tilde{\vec u})^n (\tilde{\vec u} - \lfp\sys f)$ for all $n \geq 0$.
    However, since $\specrad{\sys f'(\tilde{\vec u})} < 1$, we have by \Cref{thm:nonnegmats} that $\lim_{n \to \infty} \sys f'(\tilde{\vec u}) = 0$, so there exists $n$ such that $\sys f'(\tilde{\vec u})^n (\tilde{\vec u} - \lfp\sys f) < \tilde{\vec u} - \lfp\sys f$, contradiction.
    It follows that $\lfp\sys f = \tilde{\vec u}$.
    As $\vec l \leq \tilde{\vec u}$ we have shown that $\vec l$ is indeed a lower bound on the lfp of $\sys f$.
    
    We now prove the second part of the theorem.
    Let $\sys f = \syscl f$ be feasible, non-singular and clean.
    A rational $\vec u$ with $\syscl f(\vec u) \allsmaller \vec u$ arbitrarily close to $\lfp\syscl f$ exists by \Cref{thm:certificates}.
    It remains to show existence of a rational co-inductive $\vec l \leq \syscl f(\vec l)$.
    The argument is analogous to the proof of \Cref{thm:certificates}, but this time for lower bounds:
    Since $\syscl f$ feasible, non-singular and clean, \Cref{thm:propsOfJacobi} and \Cref{thm:nonnegmats} yield that there exists $\vec v \allgreater \vec 0$ such that $\jacat{\syscl f}{\lfp\syscl f} \vec v \allsmaller \vec v$.
    Choose such a $\vec v \leq \lfp\syscl f$.
    By the variant of Taylor's theorem,
    \begin{align}
        \lfp\syscl f - \vec v
        ~\allsmaller~
        \lfp\syscl f - \jacat{\syscl f}{\lfp\syscl f} \vec v
        ~\leq
        ~\syscl f(\lfp\syscl f - \vec v)
    \end{align}
    With $\vec l = \lfp\syscl f - \vec v$ it follows that $\vec l \allsmaller \syscl f(\vec l)$.
    Moreover, $\vec v \allgreater \vec 0$ can be chosen arbitrarily small, and hence $\vec l$ arbitrarily close to $\lfp\syscl f$.
    Finally, by continuity of $\syscl f$, the inequality $\vec l' \allsmaller \syscl f(\vec l')$ holds for all $\vec l'$ in an open ball around $\vec l$; such a ball contains rational vectors.
\end{proof}

\section{cPAST and Non-singular PPS}
\label{sec:cpast}

In this section we characterize non-singularity in terms of expected runtimes.
In a nutshell, the main result is that the (cleaned-up) fundamental PPS of a pPDA $\pda$ is non-singular iff $\pda$ is cPAST.
We begin by characterizing the 1st termination moments as the solution of a finite linear equation system.

\begin{restatable}[System of termination moments]{lemma}{momenteqs}
    \label{thm:momenteqs}
    Let $\pda = \pdainit$ be a pPDA and let $(\vetriple{p}{Z}{q})_{pZq \in \pdastates {\times} \abstack {\times} \pdastates}$ be $\exnonnegreals$-valued variables.
    Then the 1st termination moments $\etriple{p}{Z}{q}$ constitute the least solution of the following \emph{linear} PPS:
    \begin{align}
        \label{eq:esystem}
        &\vetriple{p}{Z}{q}
        ~=~
        \triple{p}{Z}{q}
        +
        \sum_{\substack{\trans{pZ}{a}{rY}}} a \cdot \vetriple{r}{Y}{q} \notag \\
        &{+}
        \sum_{\substack{\trans{pZ}{a}{rYX} \\ t \in \pdastates}} a {\cdot} \big(  \vetriple{r}{Y}{t} {\cdot} \triple{t}{X}{q} + \triple{r}{Y}{t} {\cdot} \vetriple{t}{X}{q} \big)
    \end{align}
    where $\triple{p}{Z}{q}$, etc., are the return probabilities of $\pda$.
\end{restatable}
\begin{proof}[Proof sketch]
    We sketch a proof based on a specific type of formal power series known as \emph{probability generating functions}, e.g.~\cite[App.~F]{kulkarni1995modeling}, \cite{wilf2005generatingfunctionology}.
    We will use the fact that \Cref{thm:pdafundamentalsys} actually holds in a more general setting where the transition weights come from a \emph{semiring} with suitable closure properties~\cite{DBLP:reference/hfl/Kuich97}.
    The \emph{semiring $\gfring{\exnonnegreals}{\gfvar}$ of formal power series} with coefficients in $\exnonnegreals$ satisfies these properties.
    Here, $\gfvar$ is a \emph{formal variable}, and the elements in $\gfring{\exnonnegreals}{\gfvar}$ are \emph{formal series} of the form $\sum_{k=0}^{\infty} a_k \gfvar^k$, with coefficients $a_k \in \exnonnegreals$.
    See~\cite{DBLP:reference/hfl/Kuich97} for details.
    
    Let $\gfvar \pda$ be the $\gfring{\exnonnegreals}{\gfvar}$-weighted PDA obtained from $\pda$ by multiplying each transition probability by $\gfvar$.
    The 0th termination moment is then a power series $\triple{p}{Z}{q}_{\gfvar\pda} = \sum_{k=0}^{\infty} a_k \gfvar^k$ where for all $k \geq 0$, the coefficient $a_k$ equals the sum of the probabilities of the proper paths of length $k$ between $pZ$ and $q\eps$.
    Consider the usual differentiation operator $\gfdiff{}$ for formal power series.
    We have
    \begin{align}
        \gfdiff{\triple{p}{Z}{q}_{\gfvar\pda}}
        ~=~
        \sum_{k=1}^{\infty} k \cdot a_k \gfvar^{k-1}
    \end{align}
    and thus by definition of the termination moments,
    \begin{align}
        \gfsubs{\gfdiff{\triple{p}{Z}{q}_{\gfvar\pda}}}{1}
        &~=~
        \etriple{p}{Z}{q}
        ~\text{, and also}
        \label{eq:diffandsubsismoment}\\
        \gfsubs{\triple{p}{Z}{q}_{\gfvar\pda}}{1}
        &~=~
        \triple{p}{Z}{q}
        ~,
        \label{eq:subsisprob}
    \end{align}
    where $\gfsubs{\cdot}{1} \colon \gfring{\exnonnegreals}{\gfvar} \to \exnonnegreals$ is the function that substitutes the formal variable $\gfvar$ for $1$ and takes the resulting limit (which may be $\infty$).
    By (the extended) \Cref{thm:pdafundamentalsys}, the triples $\triple{p}{Z}{q}_{\gfvar\pda}$ are the least solution in $\gfring{\exnonnegreals}{\gfvar}$ of the following equation system with variables $\vtriple{p}{Z}{q}$:
    \begin{multline}
        \label{eq:gf}
        \vtriple{p}{Z}{q}
        ~=~
        \gfvar \cdot \Big( \sum_{\trans{pZ}{a}{q\eps}} a
        ~+~
        \sum_{\trans{pZ}{a}{rY}} a \cdot \vtriple{r}{Y}{q} \\
        ~+~
        \sum_{\trans{pZ}{a}{rYX}} a \cdot \sum_{t \in \pdastates} \vtriple{r}{Y}{t} \cdot \vtriple{t}{X}{q}
        \Big)
        ~.
    \end{multline}
    \emph{Differentiating} all equations in this system using the product rule yields the following \emph{additional} equations:
    \begin{multline}
        \label{eq:gfdiff}
        \partial\vtriple{p}{Z}{q}
        ~=~
        \vtriple{p}{Z}{q}
        ~+~
        \gfvar \cdot \Big( 
        \sum_{\trans{pZ}{a}{rY}} a \cdot \gfdiff{\vtriple{r}{Y}{q}} \\
        {+}
        \sum_{\trans{pZ}{a}{rYX}} a {\cdot} \sum_{t \in \pdastates} ( \gfdiff{\vtriple{r}{Y}{t}} {\cdot} \vtriple{t}{X}{q} {+} \vtriple{r}{Y}{t} {\cdot} \gfdiff{\vtriple{t}{X}{q}} )
        \Big)
    \end{multline}
    where $\gfdiff{\vtriple{p}{Z}{q}}$, $\gfdiff{\vtriple{r}{Y}{t}}$, etc., are fresh variable symbols.
    It can be shown that the power series $\gfdiff{\triple{p}{Z}{q}_{\gfvar\pda}}$ and $\triple{p}{Z}{q}_{\gfvar\pda}$ constitute the least solution of the overall system \eqref{eq:gf} $\land$ \eqref{eq:gfdiff}.
    Finally, since the substitution operation $\gfsubs{\cdot}{1}$ is a semiring-homomorphism we may apply it to the equations in \eqref{eq:gfdiff}, and use \eqref{eq:diffandsubsismoment} and \eqref{eq:subsisprob} to obtain \eqref{eq:esystem}
\end{proof}
\begin{remark}
    It is also possible to prove \Cref{thm:momenteqs} using the equation system from~\cite{DBLP:conf/lics/EsparzaKM05} for the conditional expected runtimes $\frac{\etriple p Z q}{\triple p Z q}$ (see \iftoggle{arxiv}{\Cref{app:momenteqs}}{\cite{arxiv}}).
    However, we believe our proof using generating functions is of interest as it allows finding equation systems for \emph{any} higher termination moment by repeatedly differentiating \eqref{eq:gf}.
\end{remark}

\begin{lemma}[System of termination moments]
    \label{thm:thelinearsystem}
    Let $\pda$ be a pPDA.
    The linear PPS \eqref{eq:esystem} can be written in matrix form as follows:
    \begin{align}
        \label{eq:thelinearsystem}
        \vec x
        ~=~
        \lfp\sys f + \jacat{\sys f}{\lfp\sys f} \vec x ~,
    \end{align}
    where $\sys f$ is the fundamental PPS of $\pda$ and the variable vector $\vec x$ comprises the symbols $\vetriple{p}{Z}{q}$.
\end{lemma}
\begin{proof}
    For a general PPS $\sys g$ indexed by $J$ and $i \in J$,
    \begin{align}
        \label{eq:gendiffsys}
        (\sys g'(\lfp\sys g) \vec x)_i
        ~=~
        \sum_{j \in J} \left(\frac{\partial}{\partial_j} \sys g_i \right) (\lfp\sys g) \cdot \vec x_j
        ~.
    \end{align}

    Using \eqref{eq:gendiffsys} and the fact that \eqref{eq:thelinearsystem} is indexed by $\pdastates{\times}\abstack{\times}\pdastates$, we can write the equation in \eqref{eq:thelinearsystem} corresponding to $\vetriple{p}{Z}{q}$ as follows:
    \begin{align}
        &\vetriple{p}{Z}{q}
        ~=~
        \triple{p}{Z}{q} +
        \sum_{\substack{p'Z'q'  \in \pdastates{\times}\abstack{\times}\pdastates}} \frac{\partial}{\partial_{\vetriple{p'}{Z'}{q'}}} \Big( \notag \\
        & \sum_{\trans{pZ}{a}{q\eps}} a \notag 
        {+}
        \sum_{\trans{pZ}{a}{rY}} a {\cdot} \vetriple{r}{Y}{q} 
        {+}
        \sum_{\substack{\trans{pZ}{a}{rYX} \\ t \in \pdastates}} a {\cdot} \vetriple{r}{Y}{t} {\cdot} \vetriple{t}{X}{q}\notag \\
        & \hspace{4.5cm} \Big) (\lfp\sys f)  \cdot \vetriple{p'}{Z'}{q'} \label{eq:inter} \\
        &\phantom{\vetriple{p}{Z}{q}}~=~
        \triple{p}{Z}{q} +
        \sum_{\trans{pZ}{a}{rY}} a \cdot \vetriple{r}{Y}{q} \notag \\
        &{+}
        \sum_{\trans{pZ}{a}{rYX}} a {\cdot} \sum_{t \in \pdastates} \big( \vetriple{r}{Y}{t} {\cdot} \triple{t}{X}{q} {+} \triple{r}{Y}{t} {\cdot} \vetriple{t}{X}{q} \big) \label{eq:final}
    \end{align}
    To transform \eqref{eq:inter} into \eqref{eq:final}, we have used (among similar arguments) that for all $rYt \in \pdastates{\times}\abstack{\times}\pdastates$,
    \begin{align*}
        &\Big( \frac{\partial}{\partial_{\vetriple{r}{Y}{t}}} a {\cdot} \vetriple{r}{Y}{t} {\cdot} \vetriple{t}{X}{q} \Big) 
        (\lfp \sys f) \cdot \vetriple{r}{Y}{t} \\
        ~=~ &  \Big( a {\cdot} \vetriple{t}{X}{q} \Big) 
        (\lfp \sys f) \cdot \vetriple{r}{Y}{t} \\
        ~=~ & a \cdot \vetriple{r}{Y}{t} {\cdot} \triple{t}{X}{q} ~. \tag{as $\lfp\sys f_{tXq} = \triple{t}{X}{q}$}
    \end{align*}
    \eqref{eq:final} is identical to \Cref{eq:esystem}.
\end{proof}

One can interpret \eqref{eq:thelinearsystem} as the pPDA variant of \eqref{eq:1stmoments}.
In fact, \eqref{eq:1stmoments} can be proved using the same generating function trick as in the proof of \Cref{thm:momenteqs}.

We need one additional lemma before showing the main result of this section.
\begin{lemma}
    \label{thm:linearpps}
    Let $\syscl f$ be a feasible and clean linear PPS of the form $\vec x = \genmat \vec x + \vec b$.
    Then the equivalent conditions from \Cref{thm:nonnegmats} apply to the matrix $\genmat$.
    In particular, $\syscl f$ is non-singular.
\end{lemma}
\begin{proof}
    Note that $\syscl f'$ is a \emph{constant} matrix.
    Since $\syscl f$ is feasible and clean, we have $\specrad{\syscl f'} < 1$ by \Cref{thm:specradleq1}, which is the second condition from \Cref{thm:nonnegmats}.
\end{proof}

\begin{mdframed}
\begin{theorem}[Non-singular iff cPAST]
    \label{thm:nonsingulariffcpast}
    Let $\pda = \pdainit$ be a pPDA with cleaned-up fundamental PPS $\syscl f$.
    Then
    $\syscl f$ is non-singular
    if and only if
    $\pda$ is cPAST.
\end{theorem}
\end{mdframed}
\begin{proof}
    The cleaned-up variant of \eqref{eq:thelinearsystem} is
    \begin{align}
        \label{eq:thelinearsystemcl}
        \vec x
        ~=~
        \lfp\syscl f + \jacat{\syscl f}{\lfp\syscl f} \vec x
        ~.
    \end{align}
    This is because $\triple p Z q = 0$ iff $\etriple p Z q = 0$.
    
    Now suppose that $\syscl f$ is non-singular, i.e., $(\idmat - \jacat{\syscl f}{\lfp\syscl f})^{-1}$ exists.
    Then \eqref{eq:thelinearsystemcl} has the unique \emph{real-valued} solution $(\idmat - \jacat{\syscl f}{\lfp\syscl f})^{-1} \lfp\syscl f$.
    In particular, all the termination moments $\etriple p Z q$ differ from $\infty$.
    
    On the other hand, if $\pda$ is cPAST, then \eqref{eq:thelinearsystemcl} is feasible.
    As \eqref{eq:thelinearsystemcl} is clean, $\jacat{\syscl f}{\lfp\syscl f}$ is non-singular by \Cref{thm:linearpps}.
\end{proof}

As claimed in \Cref{fig:overview}, \Cref{thm:certificates} and \Cref{thm:nonsingulariffcpast} together imply that arbitrarily tight rational-valued inductive upper bounds on the lfp of the fundamental system $\sys f$ of a pPDA $\pda$ exist if $\pda$ is cPAST.
\section{Certifying PAST}
\label{sec:past}

In this section we are concerned with certificates for positive almost-sure termination (PAST; see Def.\ \ref{def:terminationNotions}).
We begin by characterizing the expected runtimes in terms of a linear equation system.

\begin{mdframed}
\begin{theorem}[System of expected runtimes]
    \label{thm:ertsys}
    Let $\pda = \pdainit$ be a pPDA and let $(\vert{p}{Z})_{pZ \in \pdastates{\times}\abstack}$ be $\exnonnegreals$-valued variables.
    The expected runtimes $\ert{p}{Z}$ constitute the least solution of the following \emph{linear} PPS:
    \begin{multline}
        \label{eq:ertsys}
        \vert{p}{Z}
        ~=~ 
        1
        ~+~
        \sum_{\trans{pZ}{a}{rY}} a \cdot \vert{r}{Y} \\
        {+}
        \sum_{\substack{\trans{pZ}{a}{rYX}}} a {\cdot} \big( \vert{r}{Y} {+} \sum_{t \in \pdastates} \triple{r}{Y}{t} {\cdot} \vert{t}{X} \big)
    \end{multline}
\end{theorem}
\end{mdframed}
\begin{proof}
    \newcommand{\newstate}{q^*}
    We give an automata theoretic proof.
    The idea is to construct a \emph{non-probabilistic} weighted PDA $\pda'$ such that for all $pZ \in \pdastates \times\abstack$, $\ert{p}{Z} = \triple{p}{Z}{\newstate}_{\pda'}$, where $\newstate$ is a distinguished state in $\pda'$.
    In other words, we express the expected runtimes of $\pda$ in terms of the 0th termination moments of a new automaton.
    
    Given a pPDA $\pda = \pdainit$, we construct the weighted PDA $\pda' = (\pdastates',\abstack,\pdatrans')$ as follows: $\pdastates' = \pdastates \cup \{\newstate\}$, and 
    \begin{align}
        \pdatrans'
        ~=~
        \pdatrans ~\cup~ &\{\trans{pZ}{1}{\newstate Z} \mid pZ \in \pdastates\times\abstack\} \notag \\
        ~\cup~ &\{\trans{\newstate Z}{1}{\newstate \eps} \mid Z \in \abstack\}
        ~.
    \end{align}
    Intuitively, $\pda'$ is obtained from $\pda$ by adding a new state $\newstate$ and weight-1 transitions from every state of $\pda$ to $\newstate$ (these new transitions are independent of the current topmost stack symbol).
    Moreover, once $\newstate$ is reached, $\pda'$ simply empties its stack.
    Note that $\pda'$ is not probabilistic anymore.
    
    Fix an initial configuration $pZ \in \pdastates\times\abstack$.
    Consider the 0th termination moment $\triple{p}{Z}{\newstate}_{\pda'} \in \exnonnegreals$ which is the sum of the weights of all proper finite paths 
    \begin{align}
        \label{eq:pathsinnewautom}
        \{pZ\}
        \,.\,
        (\pdastates {\times} \abstack^+)^*
        \,.\,
        (\{\newstate\} \times \abstack^+)
        \,.\,
        \{\newstate\eps\}
    \end{align}
    between $pZ$ and $\newstate\eps$ in $\mathfrak{T}_{\pda'}$.
    It follows immediately from the construction of $\pda'$ that there is a weight-preserving bijection between the paths \eqref{eq:pathsinnewautom} and the paths in $\mathfrak T _\pda$ that have not yet reached a configuration with empty stack, i.e., the paths 
    \begin{align}
        \label{eq:paths}
        \{pZ\}
        \,.\,
        (\pdastates {\times} \abstack^+)^*
        ~.
    \end{align}
    As explained below Def.~\ref{def:ert}, the sum of the probabilities of the paths \eqref{eq:paths} is exactly $\ert{p}{Z}_{\pda}$.
    Thus, $\triple{p}{Z}{\newstate}_{\pda'} = \ert{p}{Z}_{\pda}$.
    
    To finish the proof we use the fact that we have access to the 0th moments $\triple{p}{Z}{\newstate}_{\pda'}$ via the fundamental PPS $\pdasys{\pda'}$ of $\pda'$.
    By construction of $\pda'$, $\pdasys{\pda'}$ contains the equations from $\pdasys{\pda}$, as well as for each $pZ \in Q\times\Gamma$ the additional equations
    \begin{align}
        &\vtriple{p}{Z}{\newstate}
        ~=~
        \vtriple{\newstate}{Z}{\newstate}
        ~+~
        \sum_{\trans{pZ}{a}{rY}} a \cdot \vtriple{r}{Y}{\newstate} \notag \\
        &\quad +
        \sum_{\trans{pZ}{a}{rYX}} a \cdot
            \big(
                \sum_{t \in \pdastates} \vtriple{r}{Y}{t} {\cdot} \vtriple{t}{X}{\newstate}
                +
                \vtriple{r}{Y}{\newstate} {\cdot} \vtriple{\newstate}{X}{\newstate}
            \big) \\
        &\vtriple{\newstate}{Z}{\newstate}
        ~=~
        1
        ~,
    \end{align}
    where the transitions under the summation symbols refer to those of $\pda$ (\emph{not} to those of $\pda'$).
    Simplifying and identifying $\vtriple{p}{Z}{\newstate}$ with $\vert{p}{Z}$ yields \eqref{eq:ertsys}.
\end{proof}

\begin{remark}
    If one assumes upfront that the pPDA $\pda$ is AST, then it is not difficult to derive equation system \eqref{eq:ertsys} from \cite[Thm 3.1/Ex. 3.3]{DBLP:conf/lics/EsparzaKM05} (see \iftoggle{arxiv}{\Cref{app:ertsysforast}}{\cite{arxiv}} for details).
    However, \Cref{thm:ertsys} also holds if $\pda$ is not AST.
    That is, if $\sum_{q \in \pdastates}\triple{p}{Z}{q} < 1$ for some $pZ \in \pdastates \times\abstack$, then $\ert{p}{Z} = \infty$ in the least solution of \eqref{eq:ertsys}.
    In other words, system \eqref{eq:ertsys} has a non-negative real-valued solution (i.e., it is feasible) iff $\pda$ is PAST.
    The direction from left to right does not easily follow from \cite{DBLP:conf/lics/EsparzaKM05}.
    However, this direction is crucial for proving/certifying PAST without explicitly proving AST beforehand (see \Cref{thm:provingpast} and \Cref{thm:certpast} further below).
    The fact that PAST can be proved directly without proving AST first is a major insight of this paper.
    
\end{remark}

System \eqref{eq:ertsys} can also be written in matrix form as
\begin{align}
    \label{eq:theMatrixM}
    \vec r
    ~=~
    \ertmat{\pda}(\lfp\sys f) \vec r + \vec 1
    ~,
\end{align}
where $\vec r = (\vert{p}{Z})_{pZ \in \pdastates{\times}\abstack}$, and $\ertmat{\pda}(\lfp\sys f)$ is a matrix that depends on the transitions of $\pda$ and the return probabilities $\lfp\sys f$; the exact entries of $\ertmat{\pda}(\lfp\sys f)$ can be read off from \eqref{eq:ertsys}.

\begin{corollary}[Proving PAST]
    \label{thm:provingpast}
    Let $\pda$ be a pPDA with fundamental PPS $\sys f$ and suppose that $\lfp\sys f \leq \vec u$.
    If $\ertmat{\pda}(\vec u) \vec r + \vec 1 \leq \vec r$ has a solution $\vec r \in \reals_{\geq 1}^{\pdastates\times\abstack}$, then $\pda$ is PAST.
    In this case, $\ert{p}{Z} \leq \vec r_{pZ}$ for all $pZ \in \pdastates{\times}\abstack$.
\end{corollary}
\begin{proof}
    We have $\ertmat{\pda}(\vec u) \geq \ertmat{\pda}(\lfp\sys f)$ entry-wise.
    Thus any $\vec r \geq \vec 0$ with $\ertmat{\pda}(\vec u) \vec r + \vec 1 \leq \vec r$ is a component-wise upper bound on the least solution $\vec r'$ of $\ertmat{\pda}(\lfp\sys f) \vec r' + \vec 1 = \vec r'$.
    But by \Cref{thm:ertsys}, the components of $\vec r'$ are exactly the expected runtimes.
\end{proof}

\Cref{thm:provingpast} allows turning any iterative scheme (e.g.~\cite{DBLP:journals/siamcomp/EsparzaKL10,tacas}) for computing a monotonically converging sequence
\begin{align}
    \vec u^{(0)} \geq \vec u^{(1)} \ldots \geq \lfp\sys f
    ~,\quad
    \lim_{i \to \infty} \vec u^{(i)} = \lfp\sys f
    ~,
\end{align}
of \emph{upper} bounds on the lfp of $\sys f$ into a semi-algorithm that detects in finite time if a given pPDA $\pda$ is PAST or non-AST, and that may diverge only if $\pda$ is AST but not PAST.
For $i = 0,1,2,\ldots$, after computing $\vec u^{(i)}$, the semi-algorithm simply checks if
\begin{enumerate}
    \item $\ertmat{\pda}(\vec u^{(i)}) \vec r + \vec 1 \leq \vec r$ is feasible in $\nonnegreals$ in which case it reports \enquote{PAST}, or else
    \item if $\sum_{t} \triple p Z t < 1$ for some $pZ \in \pdastates{\times}\abstack$ in which case is reports \enquote{non-AST}.
    \item If neither (a) nor (b) apply, then continue with $i +1$.
\end{enumerate} 
The correctness of this algorithm is straightforward with \Cref{thm:provingpast}.
Termination in the PAST case follows because $\vec u^{(i)}$ is eventually sufficiently close to $\lfp\sys f$ such that $\ertmat{\pda}(\vec u^{(i)}) \vec r + \vec 1 \leq \vec r$ has a solution in $\nonnegreals$.

\begin{mdframed}
\begin{theorem}[Certifying PAST]
    \label{thm:certpast}
    Let $\pda = \pdainit$ be a pPDA with fundamental PPS $\sys f$.
    
    $\pda$ is PAST
    if and only if 
    $\exists \vec u \in \nonnegrats^{\pdastates\times\abstack\times\pdastates}, \vec r \in \rats_{\geq 1}^{\pdastates\times\abstack} \colon \sys f(\vec u) \leq \vec u$ and $\ertmat{\pda}(\vec u) \vec r + \vec 1 \leq \vec r$,
    where $\ertmat{\pda}$ is the matrix from \eqref{eq:theMatrixM}.
\end{theorem}
\end{mdframed}
\begin{proof}
    Assume that $\pda$ is PAST.
    Recall from \Cref{thm:terminationRelations} that PAST implies cPAST.
    By \Cref{thm:nonsingulariffcpast}, $\syscl{f}$, the cleaned-up version of $\sys f$, is non-singular.
    By \Cref{thm:certificates}, there exists $\vec u' \in \rats_{> 0}^{J'}$, where $J' \subseteq \pdastates{\times}\abstack{\times}\pdastates$ is the index set of $\syscl{f}$, such that $\syscl f(\vec u') \allsmaller \vec u'$.
    Moreover, such a $\vec u'$ can be chosen arbitrarily close to $\lfp\syscl{f}$.
    We can extend $\vec u'$ to a vector indexed by $\vec u$ in $\pdastates{\times}\abstack{\times}\pdastates$ by setting the components that belong to variables not present in $\syscl f$ to zero; we then still have $\sys f(\vec u) \leq \vec u$.
    Since $\pda$ is PAST, $\ertmat{\pda}(\lfp\sys f) \vec r + \vec 1 = \vec r$ is feasible by \Cref{thm:ertsys}, i.e., it has a solution in $\reals_{\geq 1}^{\pdastates\times\abstack}$.
    The entries of $\ertmat{\pda}(\bullet)$ depend continuously on the argument $\bullet$, so we can choose $\vec u$ sufficiently close to $\lfp\sys f$ such that $\ertmat{\pda}(\vec u) \vec r + \vec 1 = \vec r$ is still feasible.
    As $\ertmat{\pda}(\vec u) \vec r + \vec 1 = \vec r$ is obviously clean, $\ertmat{\pda}(\vec u)$ is non-singular by \Cref{thm:linearpps}.
    Since $\ertmat{\pda}(\vec u)$ has rational entries, the unique solution of $\vec r = (\idmat - \ertmat{\pda}(\vec u))^{-1}\vec 1$ is thus also rational.
    
    The other direction follows from \Cref{thm:provingpast} using $\sys f(\vec u) \leq \vec u \implies \lfp\sys f \leq \vec u$.
\end{proof}

Our results also admit a simple proof of the following result which is already known in the literature~\cite[Chap.~4]{DBLP:phd/ethos/Wojtczak09}, \cite[Thm.~3]{DBLP:conf/icalp/EtessamiWY08}.
A \emph{pBPA} is a pPDA with  $|\pdastates| = 1$.

\begin{corollary}[PAST in pBPA]
    \label{thm:pbpa}
    Let $\pda$ be a pBPA
    with fundamental PPS $\sys f$.
    Then $\pda$ is PAST iff $\vec r = \vec 1 + \sys f'(\vec 1) \vec r$ has a unique solution $\vec r \in \rats_{\geq 1}^{\abstack}$.
    In this case, $\ert{}{Z} = \vec r_{Z}$ for all $Z \in \abstack$.
    As a consequence, deciding PAST for pBPA can be done in $\P$, and the bit complexity of the expected runtimes is polynomial in the encoding size of $\pda$.
\end{corollary}
\begin{proof}
    First note that since $\pda$ has a single state, the termination moments and expected runtimes can be written as $\mathsf{E}^k[Z]$ and $\ert{Z}{}$, respectively, where $Z \in \abstack$.
    Moreover, the system \eqref{eq:esystem} of 1st moments collapses to
    \begin{align}
        \label{eq:pBPAmoments}
        \mathsf{E}\langle Z \rangle
        {=}
        [Z]
        {+}
        \sum_{\trans{Z}{a}{Y}} a {\cdot} \mathsf{E}\langle Y \rangle
        {+}
        \sum_{\trans{Z}{a}{YX}} a {\cdot} (\mathsf{E}\langle Y \rangle {\cdot} [X] {+} [Y] {\cdot} \mathsf{E}\langle X \rangle)
    \end{align}
    and the system \eqref{eq:ertsys} of expected runtimes collapses to
    \begin{align}
        \label{eq:pBPAert}
        \vert{Z}{}
        {=}
        1
        {+}
        \sum_{\trans{Z}{a}{Y}} a {\cdot} \vert{Y}{}
        {+}
        \sum_{\trans{Z}{a}{YX}} a {\cdot} (\vert{Y}{} {+} [Y] {\cdot} \vert{X}{})
        .
    \end{align}
    
    Now suppose $\pda$ is PAST.
    Then $\pda$ is also AST, so for all $Z \in \abstack$, $[Z] = 1$, i.e., $\lfp\sys f = \vec 1$.
    Then \eqref{eq:pBPAmoments} and \eqref{eq:pBPAert} are just the \emph{same} equation system, so $\mathsf{E}[Z] = \ert{Z}{}$ for all $Z \in \abstack$.
    By \Cref{thm:thelinearsystem}, \eqref{eq:pBPAmoments} can be written in matrix form as $\vec x = \vec 1 + \sys f'(\vec 1) \vec x$.
    The latter is a linear, clean and feasible PPS, so by \Cref{thm:linearpps}, it has the unique solution $\vec x = (\idmat - \sys f'(\vec 1))^{-1} \vec 1$.
    Since $\sys f'(\vec 1)$ is has rational entries, it follows from standard linear algebra that the solution is rational and can be computed in polynomial time.
    In particular, the rational numbers in the solution can be encoded with polynomially many bits.
    
    For the other direction suppose that $\vec r = \vec 1 + \sys f'(\vec 1)$ is feasible.
    But then $\eqref{eq:pBPAert}$ is also feasible, so $\pda$ is PAST.
\end{proof}

\section{Complexity of Certificates}
\label{sec:complexity}

Our results so far have been \emph{qualitative}, i.e., we have shown existence of certificates but we have not yet discussed how large (the rational numbers in) these certificates are.
This is addressed in the present section.

From \Cref{thm:nonsingulariffcpast} it follows that $\specrad{\jacat{\pdasyscl{\pda}}{\lfp\pdasyscl{\pda}}} < 1$ iff $\pda$ is cPAST.
The next lemma is a quantitative generalization:
\begin{lemma}[Bound on spectrum]
    \label{thm:specradbound}
    Let $\pda$ be a pPDA with cleaned-up fundamental PPS $\syscl f$, and suppose that $\pda$ is cPAST.
    Let $C = \max_{pZq} \frac{\etriple p Z q}{\triple p Z q} \geq 1$ be the maximal conditional expected runtime of $\pda$, and $\lfp\syscl f_{\min} = \min_{pZq} \triple{p}{Z}{q}$ be the minimum non-zero return probability.
    Then there exists $\vec v \allgreater \vec 0$, $\norm{\vec v}_1 =1$, s.t. $\jacat{\syscl f}{\lfp\syscl f} \vec v \leq (1 - C^{-1}) \vec v$, and $\vec v_{\min} \geq \frac{\lfp\syscl f_{\min}}{C |\pdastates|^2|\abstack|}$.
    
    As a consequence, $\specrad{\jacat{\syscl f}{\lfp\syscl f}} \leq (1 - C^{-1})$.
\end{lemma}
\begin{proof}
    Let $\vec v$ be the vector of positive termination moments, i.e., the unique solution of $\vec v = \lfp\syscl{f} + \syscl{f}'(\lfp\syscl f)\vec v \allgreater \vec 0$.
    Then
    \begin{align}
        \label{eq:c}
        \syscl{f}'(\lfp\syscl f)\vec v
        ~=~
        \vec v - \lfp\syscl{f} 
        ~\leq~
        c \vec v
    \end{align}
    where 
    \begin{align*}
        c
        &~=~
        \max_{pZq} \frac{\etriple p Z q - \triple p Z q}{\etriple p Z q}
        ~=~
        1 - \min_{pZq} \frac{\triple p Z q}{\etriple p Z q} \\
        &~=~
        1 - \left( \max_{pZq} \frac{\etriple p Z q}{\triple p Z q} \right)^{-1}
        ~=~
        1 - C^{-1}
        ~.
    \end{align*}
    By \Cref{thm:nonnegmatsmiscfacts}, \eqref{eq:c} implies that $c = 1 - C^{-1}$ is an upper bound on $\specrad{\syscl{f}'(\lfp\syscl f)}$.
    We can normalize $\vec v$ (in the $\norm{\cdot}_1$ norm) as follows:
    $\vec v' = \frac{\vec v}{\sum_{pZq} \etriple p Z q}$, $\norm{\vec v'} = 1$.
    Then
    \begin{align*}
        \vec v'_{\min}
        &~\geq~
        \frac{\min_{pZq} \etriple p Z q}{|\pdastates|^2 |\abstack| \max_{pZq}\etriple p Z q} \\
        &~\geq~
        \frac{\min_{pZq} \triple p Z q}{|\pdastates|^2 |\abstack| \max_{pZq} \frac{\etriple p Z q}{\triple p Z q}}
        ~=~
        \frac{\lfp\syscl f_{\min}}{|\pdastates|^2 |\abstack| C}
        ~.
    \end{align*}    
\end{proof}

The bound from \Cref{thm:specradbound} on the spectral radius $\specrad{\jacat{\syscl f}{\lfp\syscl f}}$ is tight:
Consider e.g. the pPDA $\pda_a$ with a single state and stack symbol and the transitions $\trans{pZ}{a}{\eps}$, $\trans{pZ}{1-a}{pZZ}$.
The fundamental PPS of $\pda_a$ is $\langle Z \rangle = (1-p)\langle Z \rangle^2 + p$.
For $a > \frac{1}{2}$, $\pda_a$ is PAST with expected runtime $C = \frac{1}{2a - 1}$, and the spectral radius is indeed $2(1-a) = 1 - C^{-1}$.

The next lemma allows bounding the side lengths of a (hyper)cube fully contained in the inductive region $\{\vec u \geq \vec 0 \mid \syscl f(\vec u) \leq  \vec u\}$ (see \iftoggle{arxiv}{\Cref{app:sidelengths}}{\cite{arxiv}} for the proof).
\begin{restatable}{lemma}{sidelengths}
    \label{thm:sidelengths}
    Let $\pda = \pdainit$ be a pPDA with cleaned-up fundamental PPS $\syscl{f}$.
    If $\pda$ is cPAST, then there exists a vector $\vec v \allgreater \vec 0$, $\norm{\vec v}_{1} = 1$, such that the inequality
    \begin{align}
        \label{eq:cubeproofgoal}
        \syscl{f}(\lfp\syscl{f} + d(\vec v + \vec a))
        ~\leq~
        \lfp\syscl{f} + d(\vec v + \vec a)
    \end{align}
    holds for all 
    \begin{align}
        \label{eq:dbound}
        0
        ~\leq~
        d
        ~\leq~
        \frac{\lfp\syscl f_{\min}}{C} \frac{1}{4|\pdastates|^2 |\abstack|(|\pdastates|^2 |\abstack| + 1) ^2}
    \end{align}
    and $\vec a \geq \vec 0$ with $\maxnorm{\vec a} \leq e$, where
    \begin{align}
        \label{eq:ebound}
        e
        ~=~
        \frac{\lfp\syscl{f}_{\min}}{C^2} \frac{1}{4 |\pdastates|^4 |\abstack|^2}
    \end{align}
    where $C = \max_{pZq} \frac{\etriple p Z q}{\triple p Z q} < \infty$ is the maximal conditional expected runtime in $\pda$, and $\lfp\syscl{f}_{\min} > 0$ is the minimal entry of $\lfp\syscl{f}$.
\end{restatable}

\begin{theorem}[Complexity of inductive upper bounds]
    \label{thm:sizeOfCerts}
    Let $\pda = \pdainit$ be a pPDA with fundamental PPS $\sys f$.
    Suppose that $\pda$ is cPAST with maximal conditional expected runtime $C$ and minimal non-zero return probability $\lfp{\syscl f}_{\min}$.
    Then for every $\eps > 0$ there exists $\vec u \in \nonnegrats^{\pdastates\times\abstack\times\pdastates}$,
    with $\sys f(\vec u) \leq \vec u$ and $\maxnorm{\lfp\sys f - \vec u} \leq \eps$, 
    such that the rational numbers in $\vec u$ can be encoded as pairs of integers with bit-complexity at most:
    \begin{align}
        \label{eq:sizeOfCerts}
       \mathcal{O}\big(\log C {+} \log\frac{1}{\lfp\syscl{f}_{\min}} {+} \log|\pdastates| {+} \log|\abstack| {+} \log \frac 1 \eps \big)
    \end{align}
\end{theorem}
\begin{proof}
    Note that a (hyper)cube in $\reals^n$ located within distance $\mathcal{O}(1)$ to the origin and side length $a > 0$ contains a rational vector whose entries can be encoded as pairs of integers with at most $\mathcal{O}(\log a)$ bits.
    \eqref{eq:sizeOfCerts} follows from \Cref{thm:sidelengths} by multiplying the bounds \eqref{eq:dbound}  and \eqref{eq:ebound}, taking logarithms, and omitting constant factors.
    $\eps$ can be taken into account by scaling $d$ in \eqref{eq:cubeproofgoal} appropriately.
\end{proof}

%
%
%


In the remainder of this section we discuss some bounds on the quantities occurring in \eqref{eq:sizeOfCerts}.
It is already known that $1/\lfp\syscl f_{\min}$ may be doubly-exponential in the size of $\pda$ (this means that $\log(1 / \lfp\syscl{f}_{\min})$ is exponential)~\cite{DBLP:conf/stacs/EtessamiY05}.
This can be seen in the following example.

\begin{example}
    \label{ex:repsquaring}
    Observe that the stochastic context-free grammar (SCFG)
    \begin{align}
        & Z_1 \to X^{2^n} \label{eq:scfgsquaring} \\
        & X \overset{\nicefrac 1 2}{\to} a \qquad X \overset{\nicefrac 1 2}{\to} b
    \end{align}
    produces the terminal string $a^{2^n}$ with probability $2^{-2^n}$.
    It is straightforward to simulate sampling this SCFG using a pPDA (the rule \eqref{eq:scfgsquaring} can be realized with $n$ stack symbols using a repeated squaring trick\cite{DBLP:conf/stacs/EtessamiY05}).
    The example can be extended to see that the expected runtimes, 1st termination moments, and conditional expected runtimes can all be doubly-exponential in magnitude in the size of $\pda$ as well; this works by sampling the grammar repeatedly until the outcome $a^{2^n}$ occurs for the first time, which takes $2^{2^n}$ trials on average.
    See \Cref{fig:family1} for a concrete pPDA implementing this behavior.
\end{example}

Given the close relationship between $\lfp\sys f$ and the 1st termination moments (and hence conditional expected runtimes) established in \Cref{thm:thelinearsystem}, it natural to ask whether very large runtimes only occur in pPDA with very small non-zero return probabilities (i.e., where $1/\lfp\syscl f_{\min}$ is very large; this is indeed the case in \Cref{ex:repsquaring}).
Perhaps surprisingly, the answer is \emph{no}.
To see this, we exhibit an example family where the large runtime is \enquote{caused} by a different phenomenon than in \Cref{fig:family1}.

\begin{theorem}
    There exists a familiy $\{\pda_n \mid n \geq 1\}$ of pPDA with distinguished initial configuration $pY$ such that $\ert{p}{Y} = \etriple p Y p \geq 2^{2^n}$ and the minimal non-zero return probability is $\lfp\syscl f_{\min} \geq c$, for a constant $c > 0$.
    The encoding size of $\pda_n$ is linear in $n$.
\end{theorem}
\begin{proof}
    $\pda_n$ is depicted in \Cref{fig:family2}.
    The overall idea of $\pda_n$ is to simulate the pBPA
    \begin{align}
        \label{eq:family2idea}
        \trans{Y}{a_n}{\eps}
        \qquad
        \trans{Y}{1-a_n}{YY}
    \end{align} 
    with expected runtime $(2a_n - 1)^{-1}$ for probabilities $a_n$ that approach $\nicefrac 1 2$ very fast from above as $n$ increases.
    
    To simulate probabilities $a_n > 1 - a_n$ that are very close to each other, i.e., both are \emph{almost} $\nicefrac 1 2$, the idea is as follows:
    Consider a symmetric two-state DTMC where the probability to stay in the current state is $\nicefrac 3 4$ and the probability to switch states is $\nicefrac 1 4$.
    One can show that for all $k \geq 0$,
    \begin{align}
        \label{eq:afterksteps}
        \begin{pmatrix}
            \nicefrac 3 4 & \nicefrac 1 4\\
            \nicefrac 1 4 & \nicefrac 3 4
        \end{pmatrix}^k
        ~=~
        \frac 1 2
        \begin{pmatrix}
        1 + 2^{-k} & 1 - 2^{-k}\\
        1 - 2^{-k} & 1 + 2^{-k}
        \end{pmatrix}
        ~.
    \end{align}
    This means that if the DTMC is started with initial distribution $\begin{pmatrix}1 & 0\end{pmatrix}$, then after $k$ steps the distribution is $\begin{pmatrix}\frac 1 2 + \frac{1}{2^{k+1}} & \frac 1 2 - \frac{1}{2^{k+1}} \end{pmatrix}$, i.e., one bit of precision is gained per step.
    
    The pPDA $\pda_n$ simulates the DTMC described above for $k = 2^n$ steps using a repeated squaring technique as in \Cref{ex:repsquaring} (states $s$ and $r$).
    From \eqref{eq:afterksteps} it follows that $a_n = \triple{s}{Z_1}{s} = \frac 1 2 + \frac{1}{2^{2^n+1}}$ and $1 - a_n = \triple{s}{Z_1}{r} = \frac 1 2 - \frac{1}{2^{2^n+1}}$. 
    Thus the expected runtime of $\pda_n$ from initial configuration $pY$ is at least $(2a_n - 1)^{-1} = 2^{2^n}$.
    
%
\end{proof}

There are also some upper bounds known on $1/\lfp\syscl f_{\min}$.
For instance, \cite[Lem.~5.13]{DBLP:journals/siamcomp/EsparzaKL10} proves a doubly-exponential upper bound for the case where $\syscl f$ is a strongly connected quadratic PPS.
For the (conditional) expected runtimes $C$ and the 1st termination moments we are not aware of any upper bound at all in the case of general pPDA (however, recall from \Cref{thm:pbpa}, that in certain sub-classes such as pBPA, the expected runtimes \emph{can} be bounded).

\begin{figure}[t]
    \centering
    \begin{tikzpicture}[every text node part/.style={align=left},thick,initial text=$\bot$,node distance=10mm and 16mm]
    \node[state,initial] (p) {$p$};
    \node[state,right=of p] (q) {$q$};
    \node[state,right=of q] (r) {$r$};
    \draw[->] (p) edge[loop below] node[below]  {$(1, X, \eps)$ \\ $(1, Z_i, \eps)$, \scriptsize{$1 {<} i {\leq} n$}} (p);
    \draw[->] (q) edge[loop below] node[right,yshift=-2mm]  {$(1, Z_{i-1}, Z_{i}Z_{i})$, \scriptsize{$1 {<} i {\leq} n$} \\ $(1, Z_n, XX)$ \\ $(\nicefrac 1 2, X, \eps)$} (q);
    \draw[->] (p) -- node[below] {$(1,\bot,Z_1\bot)$} (q);
    \draw[->] (q) -- node[above] {$(1,\bot,\eps)$} (r);
    \draw[->] (q) edge[bend right] node[above] {$(\nicefrac 1 2,X,\eps)$} (p);
    \end{tikzpicture}
    \caption{A family of PAST pPDA with doubly exponential expected runtime $\ert{p}{\bot} \in \Theta(2^{2^n})$. The minimal non-zero return probability is $\lfp\syscl f_{\min} = \triple{q}{Z_1}{q} = 2^{-2^n}$. Transitions not reachable from $p\bot$ are omitted.}
    \label{fig:family1}
\end{figure}

%

\begin{figure}[t]
    \centering
    \begin{tikzpicture}[every text node part/.style={align=left},thick,initial text=$Y$,node distance=10mm and 16mm]
    \node[state,initial] (p) {$p$};
    \node[state,below=of q] (r) {$r$};
    \node[state,above=of q] (s) {$s$};
    
    \draw[->] (s) edge[loop right] node[right]  {$(1, Z_{i-1}, Z_{i}Z_{i})$, \scriptsize{$1 {<} i {\leq} n$} \\ $(1, Z_n, XX)$ \\ $(\nicefrac 3 4, X, \eps)$} (s);
    \draw[->] (s) edge[bend right] node[right]  {$(\nicefrac 1 4, X, \eps)$} (r);
    \draw[->] (r) edge[loop right] node[right]  {$(1, Z_{i-1}, Z_{i}Z_{i})$, \scriptsize{$1 {<} i {\leq} n$} \\ $(1, Z_n, XX)$ \\ $(\nicefrac 3 4, X, \eps)$} (r);
    \draw[->] (r) edge[bend right=55] node[right]  {$(\nicefrac 1 4, X, \eps)$} (s);
    \draw[->] (p) -- node[sloped,above] {$(1,Y,Z_1Y)$} (s);
    \draw[->] (s) edge[bend right=45] node[above left] {$(1,Y,\eps)$} (p);
    \draw[->] (r) edge[bend left=45] node[below left] {$(1,Y,YY)$} (p);
    \end{tikzpicture}
    \caption{A family of PAST pPDA with doubly exponential expected runtime $\ert{p}{Y} \in \Theta(2^{2^n})$ and bounded minimal return probabilities $\lfp\syscl f_{\min} \geq \nicefrac 1 4$. Transitions not reachable from $pY$ are omitted.}
    \label{fig:family2}
\end{figure}

%
%
%
\section{Conclusions and Open Questions}
\label{sec:conclusion}

We have shown that the (large) class of pPDA that are cPAST admits easy-to-check certificates for two-sided bounds on the return probabilities $\triple p Z q$.
On the equation system level, we have linked cPAST to the non-singularity of the system's Jacobi matrix $\idmat - \syscl f'(\lfp\syscl f)$.
Moreover, we have presented a novel characterization of PAST and showed that it can always be certified (provided the pPDA at hand is in fact PAST).

Based on our results we conjecture the following:
\emph{Almost all pPDA are cPAST}, and in particular, \emph{almost all pPDA that are AST are also PAST}.
This is because for a random pPDA $\pda$ (with real-valued weights sampled from continuous intervals), the matrix $\syscl f'(\lfp\syscl f)$ is essentially a random matrix, and random matrices are almost-surely non-singular.
We leave an appropriate formalization of this result for future work.
\cite{tacas} provides some empirical evidence that singular $\idmat - \syscl f'(\lfp\syscl f)$ occur indeed very rarely in practice.
Also note that \emph{all} finite DTMC are cPAST, and it would be interesting to see how this evolves as one ascends to even higher levels in the generative Chomsky hierarchy~\cite{icard2020calibrating}.

Further open problems include finding a general upper bound on the expected runtimes, and on the exact complexity of computing certificates.


\section*{Acknowledgment}
We thank the anonymous reviewers for their detailed feedback which led, among many others improvements, to a significant simplification of \Cref{thm:lower}.



\bibliographystyle{IEEEtran}
\bibliography{references}
%
%
%

\iftoggle{arxiv}{
    \onecolumn
    \appendix
    \subsection{Proof of \Cref{thm:terminationRelations}}
\label{app:terminationRelations}
\terminationRelations*
\begin{proof}
    Suppose that the DTMC is $\mc = (\mcstates, \mcmat)$ and the target is $F \subseteq \mcstates$.
    The only non-trivial implications are the following ones:
    \begin{itemize}
        \item PAST $\implies$ AST: Shown in main text.
        \item PAST $\implies$ cPAST. Since PAST $\implies$ AST we know that the vector of termination probabilities is $\vec\mu^0 = \vec 1$. Thus by \eqref{eq:1stmoments} and \eqref{eq:lesert}, $\ert{s}{{\mid}F} = \mathsf{E}[s{\mid}F]$ for all $s \in \mcstates$.
        \item AST $\land$ cPAST $\implies$ PAST: Again, since we have AST we know that the vector of termination probabilities is $\vec\mu^0 = \vec 1$, and so by \eqref{eq:1stmoments} and \eqref{eq:lesert}, $\ert{s}{{\mid}F} = \mathsf{E}[s{\mid}F]$ for all $s \in \mcstates$.
        But since we have cPAST as well, the termination moments $\mathsf{E}[s{\mid}F]$ are all finite. Therefore the expected runtimes $\ert{s}{{\mid}F}$ are all finite as well.
    \end{itemize}
    For the non-implications, consider the following family of DTMC $\mc_p = (\mcstates,\mcmat)$ with $p \in (0,1)$:
    $\mcstates = \mathbb{N}$, $F = \{0\}$, and for all $i\geq 1$, $\mcmat_{i,i+1} = (1-p)$ and $\mcmat_{i, i-1} = p$.
    \begin{itemize}
        \item $\mc_{1/2}$ is AST but neither PAST nor cPAST.
        \item $\mc_{1/3}$ is cPAST but not AST (and thus also not PAST).
    \end{itemize}
\end{proof}

\subsection{Proof of variant of \Cref{thm:taylor}}
\label{app:taylor}
We show that for a PPS $\sys f$ and $\vec u, \vec v \geq \vec 0$ such that $\vec u - \vec v \geq 0$,
\begin{align}
\sys f(\vec u) - \jacat{\sys f}{\vec u} \vec v
~\leq~
\sys f(\vec u - \vec v)
~\leq~
\sys f(\vec u) - \jacat{\sys f}{\vec u - \vec v} \vec v
~.
\end{align}

The following proof is a straightforward adaption of the proof of \cite[Lemma 2.3]{DBLP:journals/siamcomp/EsparzaKL10}.
It suffices to show the claim for a multivariate polynomial $\sys f = f$ with non-negative coefficients.
Let $g(t) = f(\vec u - t \vec v)$.
We have
\begin{align}
    f(\vec u - \vec v)
    ~=~
    g(1)
    ~=~
    g(0) + \int_0^1 g'(s) ds
    ~=~
    f(\vec u) + \int_0^1 -f'(\vec u - s \vec v) \vec v ds
\end{align}
For all $s \in [0,1]$, as $f'$ has non-negative cofficients and $\vec u - \vec v \geq \vec 0$, it holds that
\begin{align}
    f'(\vec u) \geq f'(\vec u -s \vec v) \geq f'(\vec u - \vec v)
\end{align}
and thus
\begin{align}
    -f'(\vec u) \leq -f'(\vec u -s \vec v) \leq -f'(\vec u - \vec v) ~.
\end{align}
The claim follows.

\subsection{Proof of \Cref{thm:momenteqs}}
\label{app:momenteqs}
\momenteqs*
\begin{proof}
    We give a straightforward proof using an equation system from~\cite[Theorem 3.1]{DBLP:conf/lics/EsparzaKM05} (we use the formulation given in \cite[Section 5.1]{DBLP:journals/fmsd/BrazdilEKK13}).
    In these works, it was shown that the conditional expected runtimes $\frac{\etriple p Z q}{\triple p Z q}$, for $pZq \in \pdastates{\times}\abstack{\times}\pdastates$ such that $\triple p Z q > 0$, satisfy the equations
    \begin{align}
        \label{eq:esparzaconditionaleqs}
        \frac{\etriple p Z q}{\triple p Z q}
        ~=~
        1+
        \frac{1}{\triple p Z q}
            \left(
                \sum_{\trans{pZ}{a}{rY}} a \cdot \triple r Y q \cdot \frac{\etriple r Y q}{\triple r Y q}
                +
                \sum_{\trans{pZ}{a}{rYX}} a \cdot \sum_{t \in \pdastates} \triple r Y t \cdot \triple t X q \cdot 
                    \left(
                        \frac{\etriple r Y t}{\triple r Y t} + \frac{\etriple t X q}{\triple t X q}
                    \right)
            \right)
            ~.
    \end{align}
    Moreover, the conditional expected runtimes constitute exactly the least solution in $\exnonnegreals$ of the above equations when interpreting the fractions as variable symbols.
    Multiplying \eqref{eq:esparzaconditionaleqs} by $\triple p Z q$ (and reinterpreting the symbols $\etriple \cdot\cdot\cdot$ as variables) we obtain \eqref{eq:esystem}.
    Note however, that in \eqref{eq:esystem} we also have equations for \emph{all} $pZq \in \pdastates{\times}\abstack{\times}\pdastates$, not just for the ones with $\triple p Z q > 0$.
\end{proof}

\subsection{Derivation of \eqref{eq:ertsys} from \cite{DBLP:conf/lics/EsparzaKM05} in the AST case}
\label{app:ertsysforast}

Let $\pda = \pdainit$ be a pPDA.
We have already explained in \Cref{app:momenteqs} how to derive the equation system for the 1st termination moments $\etriple p Z q, pZq \in \pdastates\times\abstack\times\pdastates$, from the results of~\cite{DBLP:conf/lics/EsparzaKM05}:

\begin{align}
    \label{eq:momenteqsrestated}
    \etriple{p}{Z}{q}
    ~=~
    \triple{p}{Z}{q}
    +
    \sum_{\substack{\trans{pZ}{a}{rY}}} a \cdot \etriple{r}{Y}{q} 
    +
    \sum_{\substack{\trans{pZ}{a}{rYX} \\ t \in \pdastates}} a {\cdot} \big(  \etriple{r}{Y}{t} {\cdot} \triple{t}{X}{q} + \triple{r}{Y}{t} {\cdot} \etriple{t}{X}{q} \big)
\end{align}
\emph{Now assume that $\pda$ is AST.}
Then for all $pZ \in \pdastates\times\abstack$, we have $\ert{p}{Z} = \sum_{q \in \pdastates} \etriple{p}{Z}{q}$, and $\sum_{q \in \pdastates} \triple p Z q =1$.
Summing over all $q \in \pdastates$ in \eqref{eq:momenteqsrestated} yields
\begin{align}
    \ert p Z
    ~=~
    1
    +
    \sum_{\substack{\trans{pZ}{a}{rY}}} a \cdot \ert r Y
    +
    \sum_{\substack{\trans{pZ}{a}{rYX} \\ t \in \pdastates}} a {\cdot} \big(  \ert r Y + \triple{r}{Y}{t} {\cdot} \ert t X \big)
\end{align}
which is the same as \eqref{eq:ertsys}.
However, we stress that \eqref{eq:ertsys} also holds in the non-AST case, which does not easily follow from~\cite{DBLP:conf/lics/EsparzaKM05}.

\subsection{Proof of \Cref{thm:sidelengths}}
\label{app:sidelengths}
\sidelengths*
\begin{proof}
    \renewcommand{\d}{d}
    \renewcommand{\sys}{\syscl{f}}
    \renewcommand{\v}{\vec v}
    \renewcommand{\a}{\vec a}
    \newcommand{\e}{e}
    \renewcommand{\dim}{D}
    \newcommand{\C}{C}
    \renewcommand{\c}{c}
    Let $\v \allgreater \vec 0$ be the vector with $\norm{\v}_1 = 1$ from \Cref{thm:specradbound} that satisfies
    \begin{align}
    \label{eq:csize0}
    \sys'(\lfp\sys) \v
    ~\leq~
    (1- \C^{-1}) \v 
    ~.
    \end{align}
    For the sake of brevity we write $\c := 1-\C^{-1}$ in this proof.
    Suppose further that $\d$ and $\a$ satisfy the preconditions \eqref{eq:dbound} and \eqref{eq:ebound}.
    Taylor expansion (\Cref{thm:taylor}) yields
    \begin{align}
    \label{eq:csize1}
    \syscl{f}(\lfp\sys + \d(\vec \v + \vec \a)) &\leq \lfp\sys + \sys'(\lfp\sys + \d(\v + \a)) \d(\v + \a)
    ~.
    \end{align}
    For the RHS of \eqref{eq:csize1} to be at most $\lfp\sys + \d(\vec \v + \vec \a)$ it suffices that 
    \begin{align}
    \label{eq:csize2}
    \sys'(\lfp\sys + \d(\v + \a)) (\v + \a)
    ~\leq~
    \v + \a
    ~.
    \end{align}
    Suppose that $\syscl f$ is indexed by $J \subseteq \pdastates\times\abstack\times\pdastates$, $|J| =: D$.
    For $i \in J$, let $\sys'_i(\vec x)$ be the $i$th column of the Jacobi matrix $\sys'(\vec x)$.
    We now view the $\vec x = \sys'_i(\vec x)$ as a collection of individual PPS.
    Applying Taylor's theorem to $\sys'_i$, we obtain
    \begin{align}
    \sys'_i(\lfp\sys + \d(\vec \v + \vec \a))
    &\leq
    \sys'_i(\lfp\sys) + \sys''_i(\lfp\sys + \d(\v + \a)) \d(\v + \a) \label{eq:csize3} \\
    &\leq \sys'_i(\lfp\sys) + 2\d\v + 2\d\dim\e\vec 1 \label{eq:csize4}
    \end{align}
    In \eqref{eq:csize3}, $\sys''_i$ is the Jacobian of the PPS $\sys_i'$.
    For \eqref{eq:csize4}, we have used that $\sys''_i$ is actually a non-negative \emph{constant} matrix whose coefficients are bounded by $2$; this is because all its components are second (partial) derivatives of polynomials with total degree at most $2$ and coefficients in $[0,1]$.
    We can now express the LHS of \eqref{eq:csize2} as follows:
    \begin{align}
    &\sys'(\lfp\sys + \d(\v + \a)) (\v + \a)
    \label{eq:csize5} \\
    ~=~&
    \sum_{i=1}^D \sys'_i(\lfp\sys + \d(\vec \v + \vec \a)) (\v_i + \a_i) \\
    ~\leq~&
    \sum_{i=1}^D \left( \sys'_i(\lfp\sys) + 2\d\v + 2\d\dim\e\vec 1 \right) (\v_i + \a_i)
    \tag{by \eqref{eq:csize4}} \\
    ~\leq~&
    (\sys'(\lfp\sys) + 2\d + 2\d\dim\e)\v + \sys'(\lfp\sys)\a + 2\d\dim\e(1 + \dim\e)\vec 1 \tag{rearranging and using $\a_i \leq \e$, $\norm{\v}_1=1$} \\
    ~\leq~&
    (\sys'(\lfp\sys) + 2\d + 2\d\dim\e)\v + 2\dim\e\vec 1 + 2\d\dim\e(1 + \dim\e)\vec 1
    \tag{using that $\sys'(\lfp\sys) \leq 2$ coefficient-wise} \\
    ~=~&
    (\sys'(\lfp\sys) + 2\d + 2\d\dim\e)\v + 2\dim\e(1+ \d(1 + \dim\e))\vec 1
    \tag{rearranging}\\
    ~\leq~&
    (\c + 2\d + 2\d\dim\e)\v + 2\dim\e(1+ \d(1 + \dim\e))\vec 1
    \label{eq:csize6}
    \end{align}
    In \eqref{eq:csize6} we have used \eqref{eq:csize0}.
    Recall from \eqref{eq:csize2} that we need to ensure that expression \eqref{eq:csize5} is at most $\v+\a$.
    We show that \eqref{eq:csize5} is in fact at most $\v$.
    For the latter, it is sufficient that expression \eqref{eq:csize6} is at most $\v$, which is the case iff
    \begin{align}
    \label{eq:csize7}
    2\dim\e(1+ \d(1 + \dim\e))
    ~\leq ~
    (1 - \c - 2\d - 2\d\dim\e) \v_{\min}
    \end{align}
    where $\v_{\min} > 0$ is the minimal component of the vector $\v$.
    It can be checked\footnote{\url{https://www.wolframalpha.com/input?i=\%7B2*D*x*\%281\%2By*\%281\%2BD*x\%29\%29\%3C\%3D\%281-C-2y-2*D*y*x\%29*V\%2C0\%3CC\%3C1\%2Cx\%3E0\%2Cy\%3E0\%2C1\%3E\%3DV\%3E0\%2CD\%3E0\%7D}} that the inequality \eqref{eq:csize7} is satisfied if
    \begin{align}
    \label{eq:inequalitieswolfram}
    \e < \frac{\v_{\min}(1-\c)}{2\dim}
    \text{ and }
    \d \leq \frac{\v_{\min}(1-\c) - 2\dim\e}{2(\dim\e + 1)(\dim \e+ \v_{\min})}
    \end{align}
    To complete the proof reconsider the bounds \eqref{eq:dbound} and \eqref{eq:ebound}.
    Recall that $c = 1 - C^{-1}$, $\vec v_{\min} \geq \frac{\lfp\syscl f_{\min}}{C |\pdastates|^2|\abstack|}$ and observe that
    \begin{align}
    \frac{\lfp\syscl{f}_{\min}}{C^2} \frac{1}{4 |\pdastates|^4 |\abstack|^2}
    ~\leq~
    \frac{\v_{\min}}{C} \frac{1}{4 |\pdastates|^2 |\abstack|}
    ~<~
    \frac{\v_{\min}(1-\c)}{2\dim}
    \end{align}
    Furthermore,
    \begin{align}
    & \frac{\lfp\syscl f_{\min}}{C} \frac{1}{4|\pdastates|^2 |\abstack|(|\pdastates|^2 |\abstack| + 1) ^2} \\
    ~\leq~& \frac{\lfp\syscl f_{\min}}{C} \frac{\v_{\min}(1-\c) - 2\dim\e}{4|\pdastates|^2 |\abstack|(|\pdastates|^2 |\abstack| + 1) ^2} \\
    ~\leq~& \v_{\min} \frac{\v_{\min}(1-\c) - 2\dim\e}{4(|\pdastates|^2 |\abstack| + 1) ^2} \\
    ~\leq~& \frac{\v_{\min}(1-\c) - 2\dim\e}{2(\dim\e + 1)(\dim \e+ \v_{\min})} 
    \tag{as $\v_{\min} \leq 1$, $|\pdastates|^2|\abstack| \geq \dim$, and $\e \leq 1$}
    \end{align}
    so the inequalities \eqref{eq:inequalitieswolfram}, and hence \eqref{eq:csize7}, \eqref{eq:csize2}, and finally \eqref{eq:cubeproofgoal} are also satisfied.
\end{proof}

}{}

\end{document}